\documentclass[11pt,a4wide]{article}

\usepackage{microtype}%if unwanted, comment out or use option "draft"

\usepackage[T1]{fontenc}
\usepackage[english]{babel}
\usepackage[utf8]{inputenc}
\usepackage{alphabeta}
\usepackage{relsize}

\usepackage{amsmath,latexsym,amsthm,amsfonts,amssymb,%
ifthen,color,wrapfig,hyperref,cite,fullpage,rotate,lmodern,aliascnt,%
datetime,graphicx,algorithmic,enumerate,enumitem,%
todonotes,cleveref,bm,subcaption,tabularx,xcolor,colortbl, xspace%
}

\usepackage{latexsym,amsthm,amsmath,color,amssymb,url,hyperref}
\usepackage[ruled,linesnumbered]{algorithm2e}

\newtheorem{lemma}{Lemma}

\newtheorem{proposition}{Proposition}
\newtheorem{observation}{Observation}

\newtheorem{definition}{Definition}
\newtheorem{theorem}{Theorem}

\usepackage{tikz}
\usepackage{wasysym,marvosym}
\input{colordvi}

\newcommand{\Gx}{\G^{\overline{x}}}

\newcommand{\Bx}{B^{\overline{x}}}
\newcommand{\width}{\mathsf{width}}

\newcommand{\pw}{\mathsf{pw}}

\newcommand{\cpw}{\mathsf{cpw}}

\newcommand{\cD}{\mathbf{D}}

\newcommand{\sft}{\mathsf{t}}
\newcommand{\sfs}{\mathsf{s}}

\newcommand{\sfr}{\mathsf{r}}

\newcommand{\sfQ}{\mathsf{Q}}

\newcommand{\sfS}{\mathsf{S}}

\newcommand{\sfR}{\mathsf{R}}

\hypersetup{
 pdftitle = {},
 colorlinks = true,
 urlcolor = black!50!red,
 linkcolor = black!50!red,
 citecolor = black!50!green,
 menucolor = {red}
 filecolor = {cyan},
 anchorcolor = {black}
 urlcolor = {magenta},
 runcolor = {cyan},
 linkbordercolor = {blue}
}

\usepackage{xcolor}

\definecolor{MidnightBlue}{rgb}{0.1,0.1,0.44}
\definecolor{Black}{rgb}{0,0, 0}
\definecolor{Blue}{rgb}{0, 0 ,1}
\definecolor{Red}{rgb}{1, 0 ,0}
\definecolor{White}{rgb}{1, 1, 1}
\definecolor{Grey}{rgb}{.6, .6, .6}
\definecolor{Mygreen}{rgb}{.0, .5, .0}
\definecolor{Yellow}{rgb}{.55,.55,0}
\definecolor{mustard}{rgb}{1.0, 0.86, 0.35}
\definecolor{applegreen}{rgb}{0.55, 0.71, 0.0}
\definecolor{darkturquoise}{rgb}{0.0, 0.81, 0.82}
\definecolor{celestialblue}{rgb}{0.29, 0.59, 0.82}
\definecolor{green-yellow}{rgb}{0.68, 1.0, 0.18}
\definecolor{crimsonglory}{rgb}{0.75, 0.0, 0.2}
\definecolor{darkmagenta}{rgb}{0.30, 0.0, 0.30}

\newcommand{\red}[1]{{\color{Red}#1}}

\def\ins{\mathsf{Ins}}

\def\Rep{\mathbf{Rep}}

\def\ie{\emph{i.e.}\xspace}
\def\sfT{\mathsf{T}}

\usepackage{amssymb,amsmath,amsthm,rotate,xspace}
\usepackage{graphicx}
\usepackage{todonotes}
\usepackage{hyperref}
\usepackage{wrapfig}

\newcommand{\tseq}{\mathsf{Tseq}}
\newcommand{\rep}{\mathsf{rep}}
\newcommand{\w}{w}

\newcommand{\remove}[1]{}

\newcommand{\sshow}[2]{\ifthenelse{\equal{#1}{0}}{#2}{}}
\newcounter{func}

\newcommand{\funref}[1]{\hyperref[#1]{f_{\ref*{#1}}}} % print a
\newcounter{con}

\newcommand{\conref}[1]{\hyperref[#1]{c_{\ref*{#1}}}} % print a
\usepackage{dsfont}

\usepackage{float}
\usepackage{tikz,pgf}
\usetikzlibrary{backgrounds}
\usetikzlibrary{calc,3d}
\usetikzlibrary{intersections,decorations.pathmorphing,shapes,decorations.pathreplacing,fit,shadows,fadings}

\usepackage{latexsym,url}

\usepackage{enumerate}
\usepackage{cite}
\usepackage{ifthen}

\newcommand{\model}{\mathsf{model}}
\newcommand{\profile}{\mathsf{profile}}
\newcommand{\sfL}{\mathsf{L}}
\newcommand{\sfK}{\mathsf{K}}
\newcommand{\sfP}{\mathsf{P}}

\newcommand{\G}{\mathbf{G}}

\newcommand{\bd}{\mathbf{bd}}
\newcommand{\cc}{\mathbf{cc}}
\newcommand{\val}{\mathbf{val}}
\newcommand{\bp}{\mathbf{bp}}

\newcommand{\Gi}{G_i}

\hypersetup{
colorlinks = true,
linkcolor = black!30!blue,
citecolor = black!30!green
}

 \usepackage{titling}
\thanksmarkseries{arabic}
\newcommand*\samethanks[1][\value{footnote}]{\footnotemark[#1]}

\title{\bf A linear fixed parameter tractable algorithm\\ for connected pathwidth\thanks{An extended abstract of this paper appeared in the proceedings of Annual European Symposium on Algorithms ({ESA})~\cite{KantePT19ali}}}

\author{\bigskip Mamadou Moustapha Kanté\thanks{Université Clermont Auvergne, LIMOS, CNRS, France.  Supported by projects DEMOGRAPH (ANR-16-CE40-0028) and ASSK  (ANR-18-CE40-0025-01). Email: \texttt{mamadou.kante@uca.fr}.} \and
  Christophe Paul\thanks{CNRS, LIRMM, Univ de Montpellier, Montpellier, France. Supported  by  projects DEMOGRAPH (ANR-16-CE40-0028), ESIGMA (ANR-17-CE23-0010) and the French-German Collaboration ANR/DFG Project UTMA (ANR-20-CE92-0027). Emails: \texttt{christophe.paul@lirmm.fr}, \texttt{sedthilk@thilikos.info}\,.}
  \and
  Dimitrios  M. Thilikos\samethanks[3]  }

\begin{document}
\date{\empty}

\maketitle

\vspace{-10mm}
\begin{abstract}
\noindent 
The graph parameter of {\sl pathwidth} can be seen as a measure 
of the topological resemblance of a graph to a path. 
A popular definition
of pathwidth is given in terms of {\sl node search} 
where we are given a system of tunnels (represented by a graph) that is contaminated by  some infectious substance and we are looking for  a search strategy that, at each step, either places a searcher on a vertex or removes a searcher from a vertex
and where an edge is cleaned  when both endpoints are simultaneously occupied by searchers. It was proved that the minimum number of searchers required for a successful 
cleaning strategy is equal to the pathwidth of the graph plus one.
Two desired characteristics for a cleaning  strategy is to be 
{\sl monotone} (no recontamination occurs) and {\sl connected}
(clean territories always remain connected). Under these two demands,  
 the number of searchers  is equivalent to a variant of pathwidth called {\em connected pathwidth}.
We prove that connected pathwidth is fixed parameter tractable, in particular
we design a $2^{O(k^2)}\cdot n$ time algorithm that checks whether 
the connected pathwidth of $G$ is at most $k.$ This resolves an open question by [{\sl Dereniowski, Osula, and Rz{\k{a}}{\.{z}}ewski, Finding small-width connected path-decompositions in polynomial time. Theor. Comput. Sci., 794:85–100, 2019}\,]. For our algorithm,  we enrich the 
{\sl typical sequence technique} that is able to deal with the connectivity demand. Typical sequences have  been introduced in [{\sl Bodlaender and Kloks. Efficient and constructive algorithms for the pathwidth and treewidth of graphs. J. Algorithms, 21(2):358–402, 1996}\,] for the design of linear parameterized algorithms for treewidth and pathwidth. While this technique has been later applied to other parameters, none of  its advancements  was able to deal with the connectivity demand, as it is a ``global’’ demand that concerns
an unbounded number of parts of the graph of unbounded size.
The  proposed extension is based on an encoding of the 
connectivity property that  is quite versatile and may be adapted so to  deliver linear parameterized algorithms for the connected variants of other width parameters as well. An immediate consequence of our result is a $2^{O(k^2)}\cdot n$ time algorithm for the monotone and connected version of the edge search number.
\end{abstract}

%
%\newpage
%
%\tableofcontents

\newpage

%------------------------------------------------------------------------------------------------------------------------
\section{Introduction}

\paragraph{Pathwidth.}
A {\em path-decomposition} of a graph $G=(V,E)$ is a sequence   ${\sf Q}=\langle B_{1},\ldots,B_{q}\rangle$ of vertex sets, called {\em bags} of ${\sf Q},$ 
such that
\begin{enumerate}
\item $\bigcup_{i\in\{1,\ldots,q\}}B_{i}=V,$ 
\item every edge $e\in E$ is a subset of  some member of ${\sf Q},$
and 
\item the {\em trace} of every vertex $v\in V,$ that is the set $\{i\mid v\in B_{i}\},$  is a set of consecutive integers.
\end{enumerate}
The {\em width} of a path-decomposition is $\max\{|B_{i}|-1\mid i\in\{1,\ldots,q\}\}$ and the {\em pathwidth}
of a graph $G,$ denoted by $\pw(G),$ is the minimum width of a path-decomposition of $G.$ 

The above definition appeared for the first time in~\cite{RobertsonS83GM_I}. Pathwidth can be seen as a measure 
of the topological resemblance of a graph to a path.%
\footnote{Or, alternatively, to a caterpillar, as aptly remarked in~\cite{Telle05tree}.}  
Pathwidth, along with its tree-analogue {\sl treewidth}, have been used as key combinatorial tools in the Graph Minors series of Robertson and Seymour~\cite{RobertsonS04GMXX} and they are omnipresent 
in both structural and algorithmic graph theory. 
Apart from the above definition, pathwidth  was also defined as the {\sl interval thickness}~\cite{KirousisP85inte} (in terms of interval graphs), as the {\sl  vertex separation number}~\cite{Kinnersley92thev} (in terms of graph layouts), as the {\sl  maximum order of a  blockage}~\cite{BienstockRST91quic} (in terms of min-max dualities -- see also~\cite{FominT01onth}), and as the {\sl  node search number}~\cite{Moehring90grap,KirousisP85inte,Bienstock89grap,BienstockS91mono} (in terms of graph searching games).

Deciding whether the pathwidth of a graph is at most $k$ is an {\sf NP}-complete problem~\cite{ArnborgCP87comp}.
This  motivated the problem of the existence, or not, of a {\em parameterized algorithm} for this problem, 
and algorithm running in
$f(k)\cdot n^{O(1)}$ time algorithm.  An affirmative
answer to this question was directly implied as a consequence of the algorithmic and combinatorial 
results of the Graph Minors series and the fact that, for every $k,$ the class of graphs with pathwidth at most $k$
is closed  under taking of minors\footnote{A graph $H$ is a {\em minor} of a graph $G$ if $H$ can be obtained by some subgraph of $G$ by contracting edges.}. 
On the negative side, this implication was purely existential. 
The challenge of {\sl constructing} an $f(k)\cdot n^{O(1)}$ time algorithm for pathwidth (as well as for treewidth) was a consequence of the classic result of Bodlaender and Kloks in~\cite{BodlaenderK96effi} (see also~\cite{LagergrenA91find,CourcelleL96equi}).
The main result in ~\cite{BodlaenderK96effi} implies a  $2^{O(k^3)}\cdot n$ time algorithm. This 
was later improved to one running in $2^{O(k^2)}\cdot n$ time by Fürer in~\cite{Furer16fast}).

\paragraph{Graph searching.} In a {\em graph searching} game, the opponents 
are a group of {\sl searchers} and an evading  {\sl fugitive}. The opponents  move  in turns in a graph.
The objective of the searchers  is to deploy a strategy of moves that leads to the capture of the fugitive. At each step of the {\em node searching game}, the searchers may either place a searcher at a vertex or remove a searcher  from a vertex.
The  fugitive resides in the edges of the graph and  is lucky, invisible, fast, and agile. 
The capture of the fugitive occurs  when searchers occupy both endpoints of the edge where he  currently resides. 
A {\em node searching strategy} is a sequence of moves of the searchers that can guarantee the 
eventual  capture of the fugitive.\footnote{An equivalent setting of graph searching is to 
see $G$ as a system of pipelines or corridors that is contaminated by some poisonous gas or some highly
infectious substance. The searchers can be seen as cleaners that deploy a decontamination 
strategy~\cite{FominT08anan,ChartrandZHHMB03grap}. The fact that the fugitive is invisible, fast, lucky, and agile permits us to
see him as being omnipresent in any edge that has not yet been cleaned.} 
The cost of a searching strategy is the maximum number of
searchers simultaneously present in the graph during the deployment of the strategy. The {\em node search number}
of a graph $G,$ denoted by ${\sf ns}(G),$  is defined as the minimum cost of a  node searching strategy. 
Node searching was defined by Kirousis and Papadimitriou in~\cite{KirousisP86sear}
who proved that the game is equivalent to its monotone variant where
search strategies are  {\em monotone} in the sense that they prevent the fugitive from pervading 
again areas from where he had been expelled. This result along with the results in~\cite{Moehring90grap,Kinnersley92thev,KirousisP85inte}, imply that,
for every graph $G,$ ${\sf ns}(G)={\sf pw}(G)+1.$

\paragraph{The connectivity issue.} 
In several applications of graph searching it is important  to guarantee secure 
communication channels between the searchers so that they can safely exchange information. This issue 
was treated for the first time in the area of distributed computing,  in particular in~\cite{BarriereFFS02capt},
where the authors considered the problem of capturing an intruder by mobile agents (acting for example as 
antivirus programs). As agents deploy their cleaning strategy, they must guarantee that, at each moment of the 
search, the cleaned territories remain connected, so to permit the safe exchange  of information
between the coordinating agents. 

The systematic study of connected graph searching was initiated in~\cite{BarriereFFFNST12conn,BarriereFST03sear}. When, in node searching, we 
demand that the search strategies are monotone and connected, we define 
 {\em monotone connected node search number}, denoted by ${\sf mcns}(G).$
 \footnote{As proved in~\cite{YangDA09swee}, under the connectivity demand, the monotone and the non-monotone versions of graph searching are not any more equivalent.}
The graph decomposition counterpart of this parameter 
was introduced by Dereniowski in~\cite{Dereniowski12from}.  He defined the {\em connected pathwidth} of a connected graph $G$, denoted by ${\sf cpw}(G),$ by considering {\em connected} path-decompositions  ${\sf Q}=\{B_{1},\ldots,B_{q}\}$
where the following additional property  is satisfied:
\begin{eqnarray*}
\mbox{\red{$\blacktriangleright$} For every $i\in\{1,\ldots,q\},$ the subgraph of $G$ induced by $\bigcup_{h\in\{1,\ldots,i\}}Β_{h}$ is {\sl connected}.}\end{eqnarray*}

As noticed in~\cite{Dereniowski12from}, for every connected graph $G,$  ${\sf mcns}(G)={\sf cpw}(G)+1$ (see also~\cite{AdlerPT21conn}).
 Notice that  the above demand  results to a break of symmetry: the fact that $\langle B_{1},\ldots,B_{q}\rangle$  is a connected path-decomposition does not imply that 
the same holds for $\langle B_{q},\ldots,B_{1}\rangle$ (while this is always the case for conventional path-decompositions). 
 This break of symmetry seems to be the source of all combinatorial particularities (and challenges) of connected pathwidth. This phenomenon was also observed in the context of connected treewidth \cite{AdlerPT21conn,MescoffPT21apol}.

\paragraph{Computing connected pathwidth.} It is easy to see that
checking whether $\cpw(G)\leq k$ is an {\sf NP}-complete problem: if we define $G^*$ 
as the graph obtained from $G$ after adding a new vertex adjacent with all the vertices of $G,$ then observe that $\pw(G)=\cpw(G^*)-1.$
 This motivates the question on the parameterized complexity of the problem.
The first  progress in this direction was done recently in~\cite{DereniowskiOR19find} by  Dereniowski, Osula, and Rz{\k{a}}{\.{z}}ewski who gave 
an $f(k)\cdot n^{O(k^2)}$ time algorithm.
In~\cite[Conjecture 1]{DereniowskiOR19find}, they  conjectured that there is a fixed parameter algorithm checking whether $\cpw(G)\leq k.$ The general question on the parameterized complexity  of the connected variants of graph search was raised as an open question by Fedor V. Fomin during 
the GRASTA 2017 workshop (see~\cite{Fomin17comp}).

A somehow dissuasive fact  towards a parameterized algorithm for connected pathwidth is that connected pathwdith is not closed under minors and therefore   it does not fit
\begin{figure}
\begin{center}
\includegraphics[height=2cm]{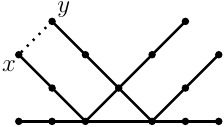}
\end{center}
\caption{A graph $G$ of connected pathwidth  $2$ with  a subgraph of connected pathwidth $3.$}
\label{dido9io}
\end{figure}

in the powerful algorithmic framework of Graph Minors (which is the case with pathwidth). The removal of an edge may increase the parameter. For instance, the connected pathwidth of the graph in \autoref{dido9io} has connected pathwidth 2 while if we remove the edge 
$\{x,y\}$
 its connected pathwidth increases to 3. On the positive side, connected pathwidth is closed under contractions (see e.g.,~\cite{AdlerPT21conn}), i.e, its value does not increase 
 when we contract edges and, moreover, the {\sf yes}-instances of the problem have bounded pathwidth, therefore they   also have bounded treewidth.
Based on these observations,  the existence of a parameterized algorithm would be implied  if we can prove that, for any $k,$ the set ${\cal Z}_{k}$ of contraction-minimal\footnote{For instance, the graph $G\setminus\{x,y\}$ from  \autoref{dido9io} belongs in ${\cal Z}_{2}.$} graphs with  connected pathwidth more than $k$ is {\sl finite}: as contraction containment can be expressed in MSO logic, one should just apply Courcelle’s theorem~\cite{Courcelle90them} to check whether some graph of ${\cal Z}_{k}$ is a contraction of $G.$ The hurdle in this direction is that 
we have no idea whether ${\cal Z}_{k}$ is finite or not.
The alternative pathway is to try  to devise a linear parameterized algorithm
by applying the algorithmic techniques that are already known for pathwidth.

\paragraph{The typical sequence technique.}
The main result of~\cite{BodlaenderK96effi} was an algorithm that, given a path-decomposition ${\sf Q}$ of $G$ of width at most $k$ and an integer $w,$ outputs, if exists, a path-decomposition of $G$ of width at most $w,$ in $2^{O(k(w+\log k))}\cdot n$ time. In this algorithm Bodlaender and Kloks introduced the celebrated {\em typical sequence technique}, a refined dynamic programming technique that encodes partial path/tree decompositions as a system of suitably compressed sequences of integers, able to encode 
all possible path-decompositions of width at most $w$. This technique was later extended/adapted for the  design of 
parametrized algorithms for numerous graph parameters such as branchwidth~\cite{BodlaenderT97cons}, linear-width~\cite{BodlaenderT98comp}, cutwidth~\cite{ThilikosSB05cutI}, carving-width~\cite{ThilikosSB00cons},  modified cutwidth, and others~\cite{ThilikosSB05cutII,BodlaenderJT20,BodlaenderFT09deri}. Also a similar 
approach was used by Lagergren in \cite{Lagergren98upp} for bounding the sizes of minor obstruction sets.
In~\cite{BodlaenderFT09deri} the typical sequence technique was viewed  as a result of {\em un-nondeterminization}: a stepwise evolution of a trivial hypothetical non-deterministic algorithm towards a deterministic parameterized algorithm. A considerable generalization of the characteristic
 sequence technique was proposed in the PhD thesis of  Soares~\cite{Pardo13purs} where this technique was implemented under the  powerful 
meta-algorithmic framework of {\em  $q$-branched $\Phi$-width}. Non-trivial extensions 
of the typical sequence technique where proposed for devising parameterized algorithms for parameters on matroids such as matroid pathwidth~\cite{Jeong0O16cons}, matroid branchwidth~\cite{Jeong0O18find}, as well as all the parameters on graphs or hypergraphs that can be expressed by them. Very recently  Bodlaender, Jaffke, and Telle in~\cite{BodlaenderJT20} suggested refinements 
of the typical sequence technique that enabled the polynomial time computation of several width parameters on directed graphs.
Finally, Bojańczyk and Pilipczuk suggested an alternative approach to the typical sequence technique, based on MSO transductions between decompositions~\cite{BojanczykP17opti}.\smallskip

Unfortunately, the above mentioned state of the art on  the typical sequence technique is unable to encompass connected pathwidth. {A reason for this is  that the connectivity demand is  a 
``global property’’ applying to {\sl every} prefix of the path-decomposition, which corresponds to 
an unbounded number of subgraphs  of arbitrary size.}

\paragraph{Our result.} In this paper we resolve {\sl affirmatively} the conjecture that checking whether ${\sf cpw}(G)\leq k$ 
is fixed parameter tractable. Our main result is the following.

\begin{theorem}
\label{main_erd}
One may construct an algorithm that given an $n$-vertex connected graph $G,$ a path-decomposition ${\sf Q}=\langle B_{1},\ldots,B_{q}\rangle$ of $G$ of width at most $k$ and an integer  $w,$ outputs 
a connected path-decomposition of $G$ of width at most $w$ or reports correctly that such an algorithm does not exist in $2^{O(k(w+\log k))}\cdot n$ time.
\end{theorem}

To design an algorithm checking whether ${\sf cpw}(G)\leq k$ we first use the algorithms
of \cite{BodlaenderK96effi} and~\cite{Furer16fast}, to build, if exists, a path 
decomposition of $G$ of width at most $k,$ in $2^{O(k^2)}\cdot n$ time. In case of a negative answer we know than ${\sf cpw}(G)> k$, otherwise  we  apply the algorithm of~\autoref{main_erd}. The overall running time 
is dominated by the algorithm of Fürer in ~\cite{Furer16fast} which is $2^{O(k^2)}\cdot n.$\smallskip

\paragraph{Our techniques.} We now give a brief description of our techniques by focusing 
on the novel issues that we introduce. This description demands some familiarity with  the typical sequence technique. Otherwise, the reader can 
go directly to the next section.

Let ${\sf Q}=\langle B_{1},\ldots,B_{q}\rangle$ be a (nice) path-decomposition of $G$ of width at most $k.$ For every $i\in[q],$ we let $\G_i=(G_i,B_i)$ be the boundaried graph  where $G_{i}=G[\bigcup_{h\in\{1,\ldots,i\}}B_{h}].$ We follow standard dynamic programming over a path-decomposition that consists  in computing a  representation of the set of partial solutions associated to $\G_i,$ which in our case are \emph{connected} path-decompositions of $\G_i$ of width at most $\w.$ The challenge is how to handle in a compact way the connectivity requirement of a path-decomposition of a graph that can be of  arbitrarily large size.

A connected path-decomposition $\sfP=\langle A_{1},\ldots,A_{\ell}\rangle$ of $\G_i$ is represented by means of a \emph{$(\G_i,\sfP)$-encoding sequence} $\sfS=\langle\sfs_1,\dots,\sfs_{\ell}\rangle.$ For 
every $j\in[\ell],$ the element $\sfs_j$ of the sequence $\sfS$ is a triple $(\bd(\sfs_j),\cc(\sfs_j),
\val(\sfs_j))$ where: $\bd(\sfs_i)=A_j\cap B_i$; $\val(\sfs_j)=|A_j\setminus B_i|$; and $\cc(\sfs_j)$ is 
the projection of the connected components of $G^{j}_{i}=G_{i}[\bigcup_{h\in\{1,\ldots,j\}}A_{h}]$ onto 
the subset of boundary vertices $B_i\cap V(G^{j}_{i}).$ To compress a \emph{$(\G_i,\sfP)$-encoding sequence} $\sfS,$ we identify a subset $\bp(\sfS)$ of indexes, called \emph{breakpoints}, such that 
$j\in\bp(\sfS)$ if $\bd(\sfs_{j-1})\neq\bd(\sfs_j)$ 
(type-$1$) or $\cc(\sfs_{j-1})\neq\cc(\sfs_{j})$ (type-$2$) or $j$ is an index belonging to a typical sequence of the integer sequence 
$\langle \val(\sfs_b),\dots,\val(\sfs_{c-1})\rangle$ where $b,c\in[\ell]$ are consecutive type-$1$ or $2$- breakpoints. We define $\rep(\sfS)$ as the induced subsequence $\sfS[\bp(\sfS)].$

The novelty in this representation is the $\cc(\cdot)$ component which is a near-partition of the subset $B_i\cap V(G^{j}_{i})$ of boundary vertices. The critical observation is that for every $j\in[\ell-1],$ $\cc(\sfs_{j+1})$ is coarser than $\cc(\sfs_j).$ This, together with the known results on typical sequences, allows us to prove that the size of $\rep(\sfS)$ is $O(k\w)$ and that the number of representative sequences is $2^{O(k(w+\log k))}.$ Finally, as in  the typical sequence technique, we define a domination relation over the set of representative sequences. The DP algorithm over the path-decomposition $\sf Q$ consists then in computing a domination set ${\bf D}_{w}(G_{i+1})$ of the representative sequences of $\G_{i+1}$ from  a domination set ${\bf D}_{w}(G_{i})$ of the representative sequences of $\G_{i}$.

The above scheme extends the current state of the art on typical sequences as it further incorporates the encoding of the connectivity property. 
While this  is indeed a ``global property’’, it appears that 
its evolution with respect to the bags of the decomposition can be controlled by the second component 
of our encoding and this is done in terms of a sequence of a gradually coarsening partitions. This establishes a dynamic programming framework
that can potentially be applied on the  connected versions of most of the parameters where the typical sequence technique was used so far. Moreover, it may be the starting point of the algorithmic study of 
parameters where other, alternative to connectivity,  global properties are imposed to the corresponding decompositions.

\paragraph{Consequences in connected graph searching.} 
The original version of graph searching was the {\em edge searching} variant, defined%
\footnote{An equivalent model was proposed independently by Petrov~\cite{Petrov82apro}. The models of Parsons and Petrov where 
different but also equivalent, as  proved by Golovach in~\cite{Golovach89equi}. 
%  П.~А.~Головач
The model of Parsons was inspired by 
an earlier paper by Breisch~\cite{Breisch67anin}, titled  {\sl ``An intuitive approach to speleotopology’’}, where the 
aim was to rescue an (unlucky)
speleologist lost in a system of caves. 
Notice that  ``unluckiness’’ cancels the speleologist's will of being rescued as,  from the searchers' point of view,  it
imposes on him/her  the status of an ``evading entity’’. As a matter of fact, the connectivity issue appears even in the first inspiring model of the search game. In a more realistic scenario, the searchers cannot ``teleport'' themselves to non-adjacent territories of the caves while  this was indeed permitted in the original setting of Parsons.}
by Parsons~\cite{Parsons78thes,Parsons78purs}, where the only differences with node searching is that a searcher can additionally 
slide along an edge and sliding is the only way to clean an edge.
The corresponding search number is called {\em edge search number}
and is denoted by  ${\sf es}(G).$ If we additionally demand that the searching strategy  is connected and monotone, then we define the  {\em monotone connected edge search number} denoted by ${\sf mces}(G).$  As proved in~\cite{KirousisP86sear}, ${\sf es}(G)={\sf pw}(G_{\sf v}),$ where $G_{\sf v}$ is the graph obtained if we subdivide twice each edge of $G.$ Applying the same reduction as in~\cite{KirousisP86sear} for  the monotone and connected setting,  one can prove that  ${\sf mces}(G)={\sf cpw}(G_{\sf v}).$
As we already mentioned, ${\sf mcns}(G)={\sf cpw}(G_{\sf v})+1.$ These two reductions imply that the result of \autoref{main_erd} holds also 
for ${\sf mcns}$ and ${\sf mces},$ i.e., the search numbers for the monotone and connected versions of both node and edge searching.

%------------------------------------------------------------------------------------------------------------------------
\section{Preliminaries and definitions}
\label{sec:prelim}

\paragraph{Sets and near-partitions.} 
{For an integer $\ell>0$, we  denote by $[\ell]$ the set $\{1,\dots, \ell\}$.}
 Let $S$ be a finite set. A \emph{near-partition} $\mathsf{Q}$ of $S$ is a family of subsets $\{X_1,\dots, X_k\}$ (with $k\leq|S|+1$) of subsets of $S$, called \emph{blocks}, such that $\bigcup_{i\in[k]} X_i=S$ and for every $1\leq i<j\leq k$, $X_i\cap X_j=\emptyset$. Observe that  a near-partition may contain several copies of the empty set. A \emph{partition} of $S$ is a near-partition with the additional constraint that if it contains the empty set, then this is the unique block.
Let $\mathsf{Q}$ be a near-partition of a set $S$ and $\mathsf{Q'}$ be a near-partition of a set $S'$ such that $S\subseteq S'$. We say that $\mathsf{Q}$ is
\emph{thinner} than $\mathsf{Q'}$, or that $\mathsf{Q'}$ is \emph{coarser} than $\mathsf{Q}$, which we denote by $\mathsf{Q}\sqsubseteq \mathsf{Q'}$, if for every block $X$ of $\mathsf{Q}$, there exists a block $X'$ of $\mathsf{Q'}$ such that $X\subseteq X'$. 
For a near-partition $\mathsf{Q}=\{X_1,\dots, X_\ell\}$ of $S$ and a subset $S'\subseteq S$, we define the  \emph{projection of $\sfQ$ onto $S'$} as the near-partition $\mathsf{Q}_{\mid S'}=\{X_1\cap S',\dots, X_\ell\cap S'\}$. Observe that if $\mathsf{Q}$ is a partition, then $\mathsf{Q}_{\mid S'}$ 
{
may not be a partition: if several blocks of $\mathsf{Q}$ are subsets of $S\setminus S’$, then $\mathsf{Q}_{\mid S'}$ contains several copies of the emptyset.
}

\paragraph{Sequences.}
Let $S$ be a set. A \emph{sequence} of elements of $S$, denoted by $\alpha=\langle a_1,\dots, a_{\ell}\rangle$, is a subset of $S$ equipped with a total ordering: for $1\leqslant i<j\leqslant \ell$, $a_i$
occurs before $a_j$ in the sequence $\alpha$. The \emph{length} of a sequence is the number of elements that it contains. 
Let $X\subseteq [\ell]$ be a subset of indexes of $\alpha$. We define the {\em subsequence} of $\alpha$ {\em induced} by $X$ as the sequence $\alpha[X]$ on the subset $\{a_i\mid i\in X\}$ such that, for $i,j\in X$, $a_i$ occurs before $a_j$ in $\alpha[X]$ if and  only if $i<j$. 
{If $\alpha=\langle a_1,\ldots,a_\ell\rangle$ and $\beta=\langle b_1,\ldots, b_{p}\rangle$ are two sequences, we let $\alpha\circ\beta$ denote the concatenation of $\alpha$ and $\beta$, \ie, $\alpha\circ \beta$ is the  sequence $\langle a_1,\ldots, a_\ell, b_1,\ldots, b_{p}\rangle$.}

The \emph{duplication} of the element $a_j$, with $j\in[\ell]$, in the sequence $\alpha=\langle a,\dots, a_\ell\rangle$ yields the  sequence
$\alpha'=\langle a_1,\dots, a_{j-1},a_j, a_j,a_{j+1}, \dots, a_\ell\rangle$ of length $\ell+1$. A sequence $\beta$ is an \emph{extension} of the sequence $\alpha$ if it is either $α$ or it results from a series of duplications on $α$. We define the set of extensions of $\alpha$ as:
${\sf Ext}(\alpha)=\{\alpha^*\mid \alpha^* \mbox{ is an extension of } \alpha\}.$

Let $\alpha=\langle a_1,\dots, a_{\ell}\rangle$ be a sequence and $\alpha^*=\langle a_1,\dots, a_p\rangle$ be an extension of $\alpha$.
If $p\leq \ell+k$, then $\alpha^*$ results from a series of at most $k$ duplications and we say that $\alpha^*$ is a \emph{$(\leq k)$-extension} of $\alpha$. 
With the definition of an extension, every element of $\alpha^*$ is a copy of some element of $\alpha$. We define the \emph{extension surjection} as a surjective function $\delta_{\alpha^*\rightarrow\alpha}:[p]\rightarrow [\ell]$ such that if $\delta_{\alpha^*\rightarrow\alpha}(j)=i$, then $a^*_j=a_i$. 
An extension surjection $\delta_{\alpha^*\rightarrow\alpha}$ is a certificate that $\alpha^*\in {\sf Ext}(\alpha)$.
Finally, we observe that if $\alpha^*\in {\sf Ext}(\alpha)$, then $\alpha$ is an induced subsequence of $\alpha^*$. 
Moreover, if $\alpha^*\in {\sf Ext}(\alpha)$ and $\beta\in {\sf Ext}(\alpha^*)$, then $\beta$ is an extension of $\alpha$.

\begin{figure}[h]
\centering
\begin{tikzpicture}[scale=1.2]
   \tikzstyle{vertex}=[fill,red,circle,minimum size=0.15cm,inner sep=0pt]
   \tikzstyle{vertex2}=[fill,black,diamond,minimum size=0.15cm,inner sep=0pt]

%------ \sfS -------
\draw[black,thin,-] (-3,3) -- (3,3);
\node[anchor=east] at (-3,3) {$\alpha$};

	%\node[vertex] (s0) at (-3,3) {};

\foreach \i in {1,...,11}{
	\node[vertex] (s\i) at (\i/2-3,3) {};
	\node[anchor=south] at (s\i) {\tiny $a_{\i}$};
	}

%------ \sfR -------
\draw[black,thin,-] (-5,1) -- (5,1);
\node[anchor=east] at (-5,1) {$\beta$};

\foreach \j in {2,3,6,7,11,13,14,18}{
	\node[vertex2] (r\j) at (\j/2-5,1) {};
	\node[anchor=north] at (r\j) {\tiny $b_{\j}$};
	}

\foreach \j in {1,4,5,8,9,10,12,15,16,17,19}{
	\node[vertex] (r\j) at (\j/2-5,1) {};
	\node[anchor=north] at (r\j) {\tiny $b_{\j}$};
	}

%----- \delta R-S -----
\draw[gray!60,thick,->,dotted] (r1) -- (s1);
\draw[gray!60,thick,->,dotted] (r2) -- (s1);
\draw[gray!60,thick,->,dotted] (r3) -- (s1);
\draw[gray!60,thick,->,dotted] (r4) -- (s2);
\draw[gray!60,thick,->,dotted] (r5) -- (s3);
\draw[gray!60,thick,->,dotted] (r6) -- (s3);
\draw[gray!60,thick,->,dotted] (r7) -- (s3);
\draw[gray!60,thick,->,dotted] (r8) -- (s4);
\draw[gray!60,thick,->,dotted] (r9) -- (s5);
\draw[gray!60,thick,->,dotted] (r10) -- (s6);
\draw[gray!60,thick,->,dotted] (r11) -- (s6);
\draw[gray!60,thick,->,dotted] (r12) -- (s7);
\draw[gray!60,thick,->,dotted] (r13) -- (s7);
\draw[gray!60,thick,->,dotted] (r14) -- (s7);
\draw[gray!60,thick,->,dotted] (r15) -- (s8);
\draw[gray!60,thick,->,dotted] (r16) -- (s9);
\draw[gray!60,thick,->,dotted] (r17) -- (s10);
\draw[gray!60,thick,->,dotted] (r18) -- (s10);
\draw[gray!60,thick,->,dotted] (r19) -- (s11);

\node[anchor=west] at (4,2) {$\delta_{\beta\rightarrow \alpha}(\cdot)$};

\end{tikzpicture}
\caption{The sequence $\beta=\langle b_1,\dots, b_{19}\rangle$ is an $(\leq 8)$-extension of the sequence $\alpha=\langle a_1,\dots, a_{11}\rangle$. The element $a_3$ has been duplicated twice in $\beta$ yielding three copies $b_{5}$, $b_{6}$,  and $b_{7}$, which are certified by  $\delta_{\beta\rightarrow \alpha}(5)=\delta_{\beta\rightarrow \alpha}(6)=\delta_{\beta\rightarrow \alpha}(7)=3$.
\label{fig:extension-surjection}}
\end{figure}
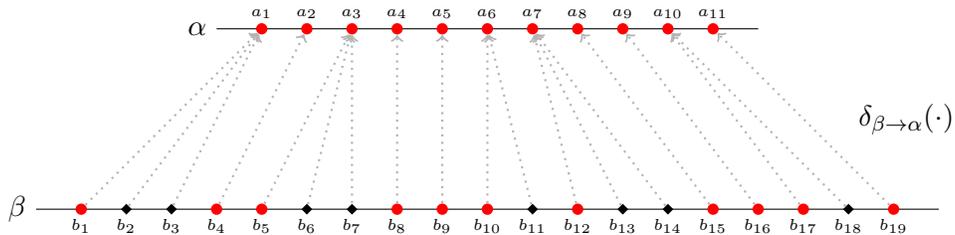

\paragraph{Graphs and boundaried graphs.}
{Given a graph $G=(V,E)$ and a vertex set $S\subseteq V(G)$, we denote by $G[S]$
the {\em subgraph} of $G$ that is {\em  induced by the vertices of} $S$, i.e., the graph  $(S,\{e\in E\mid e\subseteq S\})$}. Also, if $x\in V$, we define $G\setminus x=G[V\setminus\{x\}]$. 
 The {\em neighborhood} of a vertex $v$ in $G$
is the set of vertices that are adjacent to $v$ in $G$ and is denoted by $N_{G}(v)$.

A \emph{boundaried graph} is a pair $\mathbf{G}=(G,B)$ such that $G$ is a graph over a vertex set $V$ and $B\subseteq V$ is a subset of distinguished vertices, called \emph{boundary vertices}. The vertices of $V\setminus B$ are called \emph{inactive} vertices. 
We say that a boundaried graph $\mathbf{G}=(G,B)$ is \emph{connected} if either $G$ is connected and $B=\emptyset$ or, in case $B\ne \emptyset$,  every connected component $C$ of $G$ contains some
boundary vertex, that is $C\cap B\neq\emptyset$.

%--------------------------------------
\subsection{Connected pathwidth.} 

A \emph{path-decomposition}  of a graph $G=(V,E)$ is a sequence $\sfP=\langle A_1,\dots, A_p\rangle$ of subsets of $V$ where:
\begin{enumerate}
\item for every vertex $x\in V$, there exists $i\in[p]$ such that $x\in A_i$;
\item for every edge $e\in E$, there exists $i\in[p]$  such that $e\subseteq A_i$;
\item for every vertex $x\in V$, the set $\mathcal{A}(x)=\{i\in P\mid x\in A_i\}$ is a subset of consecutive integers.
\end{enumerate}
Hereafter, the subsets $A_i$'s (for $i\in[p]$) are called the \emph{bags} of the path-decomposition $\sfP$ and the set $\mathcal{A}(x)$ is the \emph{trace} of $x$ in $\sfP$.
 The \emph{width} of a path-decomposition is  $\mathsf{width}(\sfP)=\max\{|A_i|-1\mid i\in [p]\}$. 
The \emph{pathwidth} of a graph $G$, denoted by $\pw(G)$, is the least width of a path-decomposition of $G$. 
Finally, for every $i\in [p]$, we define $V_i=\bigcup_{j\leq i} A_j$ and $G_i=G[V_i]$.

A path-decomposition $\sfP=\langle A_1,\dots, A_{\ell}\rangle$ of a graph $G$ is \emph{nice} if $|A_1|=1$ and for every $1< i\leq p$, the symmetric difference $A_{i-1}\vartriangle A_i$ has
size one. We distinguish two types of bags:
\begin{itemize}
\item if $A_{i-1}\subset A_i$ ($1<i\leq p$), then $A_i$ is an \emph{introduce} bag ($A_1$ is also defined as an introduce bag);
\item if $A_i\subset A_{i-1}$ ($1<i\leq p$), then $A_i$ is a \emph{forget} bag. 
\end{itemize}

It is well-known that any path-decomposition can be turned in {linear} time into a nice path-decompo\-si\-tion of same width {(see e.g.,~\cite{BodlaenderK96effi})}.

\begin{definition}[Connected path-decomposition]
  A path-decomposition $\sfP=\langle A_1,\dots, A_{\ell}\rangle$ of a {connected} graph $G$ is \emph{connected} if, for every $i \in [\ell]$, the subgraph $G_i$ is
  connected. The \emph{connected pathwidth}, denoted by $\cpw(G)$, is the smallest width of a connected path-decomposition of $G$.
\end{definition}

Let us notice that if $\sfP=\langle A_1,\dots, A_{\ell}\rangle$ is a path-decomposition of a graph $G$, then $\sfP'=\langle A_p,\dots, A_1\rangle$ is also a path-decomposition of $G$. But the fact that $\sfP$ is a connected path-decomposition does not imply that $\sfP'$ is a connected path-decomposition.

\begin{observation} \label{obs:cpw}
For every graph $G$, $\pw(G)\leqslant \cpw(G)$.
\end{observation}

Let $\sfP$ be a path-decomposition of a graph $G=(V,E)$. Then for every subset $B\subseteq V$, $\sfP$ is a path-decomposition of the connected boundaried graph $\G=(G,B)$. The definition of a connected path-decomposition also naturally extends to boundaried graphs as follows.

\begin{definition}[Connected path-decomposition of a boundaried graph]
Let  $\sfP=\langle A_1,\dots, A_{\ell}\rangle$ be a path-decomposition of the boundaried graph $\mathbf{G}=(G,B)$. Then $\sfP$ is \emph{connected} if, for every $i\in [p]$, the boundaried graph $\Gi=(G_i,V_i\cap B)$ is connected.
\end{definition}

Let  $\sfP=\langle A_1,\dots, A_{\ell}\rangle$ be a path-decomposition of $\G=(G,B)$. If $x$ is a vertex of $G$, then 
$\langle A_1\setminus\{x\},\dots, A_{\ell}\setminus\{x\}\rangle$, is a path-decomposition of $(G\setminus x, B\setminus \{x\})$. Notice that we may have a bag $A_i$ of $P$ such that $A_i\setminus \{x\}=\emptyset$, but this does not contradict the definition of path-decomposition.  However, the fact that $\sfP$ is a connected path-decomposition does not imply that $\langle A_1\setminus\{x\},\dots, A_{\ell}\setminus\{x\}\rangle$ is. The following lemma establishes a condition for the vertex $x$ to satisfy so that its removal preserves connectivity.

\begin{lemma}\label{lem:proj-boundaried} 
Let $\sfP=\langle A_1,\dots, A_{\ell}\rangle$ be a  connected  path-decomposition of the  connected boundaried graph $(G,B)$.
If $x$ is a vertex of $B$ such that $N_G(x)\subseteq B$, then $\langle A_1\setminus\{x\},\dots, A_{\ell}\setminus\{x\}\rangle$ is a connected path-decomposition of $(G\setminus x,B\setminus \{x\})$.
\end{lemma}

\begin{proof}
  As already observed, $\sfP^{\overline{x}}=\langle A_1\setminus\{x\},\dots, A_{\ell}\setminus\{x\}\rangle$ is a path-decomposition of $G\setminus x$. Suppose
  that $[f,l]$ with $1\leq f\leq l\leq \ell$ is the trace of $x$ in $\sfP$.  As for every integer $i<f$ (supposing that $1<l$), the
  boundaried graph $(G_i\setminus x,(V_i\cap B)\setminus\{x\})$ is equal to $(G_i,V_i\cap B)$ and is thereby connected. So, let us consider
  an integer $i$ such that $f\leq i$. Let $C_x$ be the connected component of $G_i$ that contains $x$. As $(G_i,V_i\cap B)$ is connected,
  every connected component of $G_i$ intersects $B$. Observe that every connected component $C$ of $G_i$ distinct from $C_x$ (if any) is a
  connected component of $G[V_i\setminus\{x\}]$ which intersects $B\setminus\{x\}$. If $C_x=\{x\}$, by the previous observations, the
  statement holds. So, let $C_1,\dots, C_s$, with $s\geq 1$, be the connected components of $G[V_i\setminus\{x\}]$ such that for every
  $j\in[s]$, $C_j\subsetneq C_x$. As $C_x\neq\{x\}$, for every $j\in[s]$, $C_j$ contains a neighbor of $x$ which by assumption belongs to
  $B\setminus\{x\}$. It follows that every connected component of $G_i\setminus x$ contains a vertex of $B\setminus \{x\}$. Thereby
  $(G_i\setminus x,(V_i\cap B)\setminus \{x\})$ is a connected boundaried graph implying that $\sfP^{\overline{x}}$ is a connected
  path-decomposition of $(G\setminus x,B\setminus \{x\})$.
\end{proof}

%-----------------------------------------------------------------------------

\subsection{Integer sequences} 
\label{tipstoj}

Let us recall the notion of \emph{typical sequences} introduced by Bodlaender and Kloks~\cite{BodlaenderK96effi}  (see also~\cite{LagergrenA91find,CourcelleL96equi}).

\begin{definition} \label{def:rep-seq}
Let $\alpha=\langle a_1,\dots, a_\ell\rangle$ be an integer sequence. The \emph{typical sequence} $\tseq(\alpha)$ is obtained after iterating the following operations, until none is possible anymore:
\begin{itemize}
\item if for some $i\in[\ell-1]$, $a_i=a_{i+1}$, then remove $a_{i+1}$ from $\alpha$;
\item if there exists $i,j\in[\ell]$ such that $i\leqslant j-2$ and $\forall h$, $i< h< j$, $a_i\leq a_h\leq a_j$ or $\forall h$,  $i< h< j$, $a_i\geq a_h\geq a_j$, then remove the subsequence $\langle a_{i+1},\dots, a_{j-1}\rangle$ from $\alpha$.
\end{itemize}
\end{definition}

As a typical sequence $\tseq(\alpha)=\langle b_1,\dots, b_i,\dots, b_r\rangle$ is a subsequence of $\alpha$, it follows that, for every $i\in[r]$, there exists $j_i\in [\ell]$ such that $b_i=a_{j_i}$. Herefater every such index $j_i$ is called a \emph{tip} of the sequence $\alpha$.

\begin{lemma}[\cite{BodlaenderK96effi}]\label{lem:size-ts}
  Let $\alpha=\langle a_1,\dots, a_\ell\rangle$ be an integer sequence. Then, $\tseq(\alpha)$ is uniquely defined. If, moreover, for every $ i\in[\ell]$, we have
  $a_i\in \{0,1,\ldots, k\}$, then the length of $\tseq(\alpha)$ is at most $2k+1$.
\end{lemma}

\begin{figure}[h]

\centering
\usetikzlibrary{shapes}
\begin{tikzpicture}[scale=1.2]
   \tikzstyle{vertex}=[fill,circle,minimum size=0.2cm,inner sep=0pt]
   \tikzstyle{bvertex}=[fill,gray,diamond,minimum size=0.2cm,inner sep=0pt]

\foreach \i in {
-10,-9,-8,-7,-6,-5,-4,-3,-2,-1,0,1,2,3,4,5,6,7,8,9,10,11
}
{\draw[gray,thin,-] (\i/2,2.5) -- (\i/2,-2);}

\draw[black,thin,->] (-5.5,-2) to (6,-2);
\draw[black,thin,->] (-5.5,-2) to (-5.5,3);

\foreach \i in {
-3,-2,-1,0,1,2,3,4,5
}
{\draw[gray,thin,-] (-5.5,\i/2) -- (5.5,\i/2);}

%---------------------------------

\foreach \i/\posx/\posy in {
0/-10/0,1/-7/3,2/-6/-1,3/-3/5,4/1/-3,5/4/4,6/7/-1,7/11/2
}
{\node[vertex] (j\i) at (\posx/2,\posy/2){};}

\foreach \i/\h in {
0/1,1/2,2/3,3/4,4/5,5/6,6/7%
}
{\draw[thin] (j\i) to (j\h);
} 

\node[anchor=east] at (-5.5,-1.5) {\scriptsize $1$};
\node[anchor=east] at (-5.5,-0.5) {\scriptsize $3$};
\node[anchor=east] at (-5.5,0.5) {\scriptsize $5$};
\node[anchor=east] at (-5.5,1.5) {\scriptsize $7$};
\node[anchor=east] at (-5.5,2.5) {\scriptsize $9$};
\node[anchor=south] at (-5.5,3) {\scriptsize $a_i$};

\node[anchor=north] at (-5,-2) {\scriptsize $1$};
\node[anchor=north] at (-4,-2) {\scriptsize $3$};
\node[anchor=north] at (-3,-2) {\scriptsize $5$};
\node[anchor=north] at (-2,-2) {\scriptsize $7$};
\node[anchor=north] at (-1,-2) {\scriptsize $9$};
\node[anchor=north] at (0,-2) {\scriptsize $11$};
\node[anchor=north] at (1,-2) {\scriptsize $13$};
\node[anchor=north] at (2,-2) {\scriptsize $15$};
\node[anchor=north] at (3,-2) {\scriptsize $17$};
\node[anchor=north] at (4,-2) {\scriptsize $19$};
\node[anchor=north] at (5,-2) {\scriptsize $21$};
\node[anchor=west] at (6,-2) {\scriptsize $i$};

\foreach \j/\posa/\posb in {
1/-9/2,2/-8/1,3/-5/1,4/-4/3,5/-2/0,6/-1/2,7/0/-1,8/2/0,9/3/3,10/5/1,11/6/2,12/8/0,13/9/0,14/10/1
}
{\node[bvertex] (x\j) at (\posa/2,\posb/2){};}

\end{tikzpicture}
\caption{The black bullets forms the typical sequence $\tseq(\alpha)=\langle 4,7,3,9,1,8,3,6\rangle$ of the sequence $\alpha=\langle 4,6,5,7,3,5,7,9,4,6,3,1,4,7,8,5,6,3,4,4,5,6\rangle$ represented by black bullets and gray diamonds.\label{fig:typical-sequence}}
\end{figure}
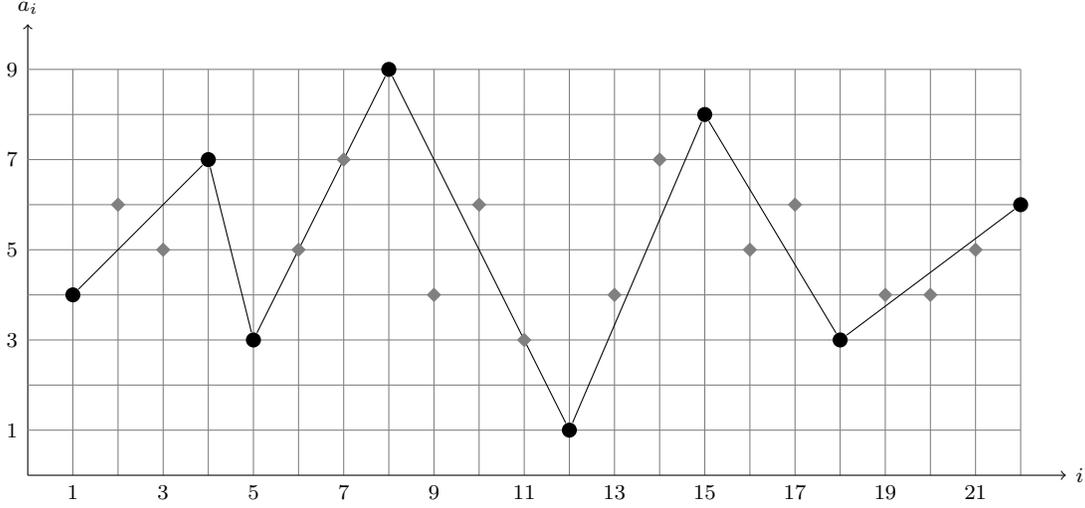

\begin{lemma}[\cite{BodlaenderK96effi}]\label{lem:nb-ts} The number of different typical sequences of integers in $\{0,1,\ldots,k\}$ is at most $\frac{8}{3}\cdot 2^{2k}$. 
\end{lemma}

 A consequence of the next lemma is that every tip of the sequence $\alpha\circ\beta$ is a tip of $\alpha$ or of $\beta$.

\begin{lemma}[\cite{BodlaenderK96effi}] \label{lem:tip-concat} 
Let $\alpha$ and $\beta$ be two integer sequences. Then, $\tseq(\alpha\circ \beta)= \tseq(\tseq(\alpha)\circ \tseq(\beta))$. 
\end{lemma}

If $\alpha$ and $\beta$ are two integer sequences of same length $\ell$, we say that $\alpha \leq \beta$ if for every $j\in[\ell]$, $a_j\leq b_j$.

\begin{definition}
Let $\alpha$ and $\beta$ be two integer sequences. Then $\alpha\preceq\beta$ if there are $\alpha^*\in {\sf Ext}(\alpha)$ and $\beta^*\in {\sf Ext}(\beta)$ such that $\alpha^*\leq\beta^*$. Whenever $\alpha\preceq\beta$ and $\beta\preceq\alpha$, we say that $\alpha$ and $\beta$ are \emph{equivalent} which is denoted by $\alpha\equiv \beta$.
\end{definition}

We summarize in the following a set of known properties concerning duplications of integer sequences and the binary relation $\preceq$ we will need.

\begin{lemma}[\cite{BodlaenderK96effi}]\label{lem:prop-useful} 
Let $\alpha$ and $\beta$ be two integer sequences.
\begin{enumerate}
  \item If $\alpha$ has length at most $\ell$, then ${\sf Ext}(\alpha)$ contains at most $2^{\ell-1}$ sequences of length $\ell$.
  \item If $\alpha^*\in {\sf Ext}(\alpha)$, then $\alpha\equiv\alpha^*$.
  \item If $\alpha^*\in {\sf Ext}(\alpha)$ and $\beta^*\in {\sf Ext}(\beta)$, then $\alpha^*\circ \beta^* \in {\sf Ext}(\alpha\circ \beta)$.
  \item If $\alpha' \preceq \alpha$ and $\beta'\preceq \beta$, then $\alpha'\circ \beta' \preceq \alpha\circ \beta$.
  \item  The relation $\preceq$ is transitive, and $\equiv$ is an equivalence relation.
  \item For every integer sequence $\alpha$, we have $\tseq(\alpha) \equiv \alpha$.  Moreover, there exist extensions $\alpha'$ and $\alpha''$ of $\tseq(\alpha)$ such that
    $\alpha' \leq \alpha \leq \alpha''$.
  \item  $\alpha\preceq\beta$ if and only if $\tseq(\alpha)\preceq\tseq(\beta)$.
    \end{enumerate}
\end{lemma}

We extend the definition of the $\leq$-relation and $\preceq$-relation on integer sequences to sequences of integer sequences. Let $\mathsf{P}=\langle \sfL_1,\dots, , L_r\rangle$ and $\mathsf{Q}=\langle \sfK_1,\dots, \sfK_r\rangle$ be two sequences of integer sequences such that for every $i\in[r]$, $\sfL_i$ and $\sfK_i$ have the same length. We say that $\sfP\leq \sfQ$ if for every $i\in[r]$, $\sfL_i\leq \sfK_i$. The set of \emph{extensions} of $\mathsf{P}$ is ${\sf Ext}(\mathsf{P})=\{\langle \sfL'_1,\dots, \sfL'_r\rangle\mid i\in[r],\sfL'_i\in {\sf Ext}(\sfL_i)\}$. Finally we say that $\sfP\preceq\sfQ$ if there exist $\sfP'\in {\sf Ext}(\sfP)$ and $\sfQ'\in {\sf Ext}(\sfQ)$ such that $\sfP'\leq \sfQ'$. If $\sfP\preceq\sfQ$ and $\sfQ\preceq\sfP$, then we say that $\sfP\equiv \sfQ$. {The relation $\equiv$ is an equivalence relation.}

%------------------------------------------------------------------------------------------------------------------------
\section{Boundaried sequences} 
\label{sec:boundaried-section}

\begin{definition}[$B$-boundaried sequence]
Let $B$ be a finite set. A \emph{$B$-boundaried sequence} is a sequence $\sfS=\langle\sfs_1,\dots, \sfs_{\ell}\rangle$ such that for every $j\in[\ell]$, $\sfs_j=(\bd(\sfs_j),\cc(\sfs_j), \val(\sfs_j))$ is defined as follows:
 \begin{itemize}
\item $\bd(\sfs_j)\subseteq B$ with the property that for every $x\in B$, the indices $j\in[l]$ such that $x\in\bd(\sfs_j)$ are consecutive; 
\item $\cc(\sfs_j)$ is a near-partition of  $\bigcup_{i\leq j} \bd(\sfs_i)\subseteq B$ with the property that for every $j<\ell$, $\cc(\sfs_j) \sqsubseteq \cc(\sfs_{j+1})$;
\item $\val(\sfs_j)$ is a positive integer.
\end{itemize}
\end{definition}

The \emph{width} of $\sfS$ is defined as $\width(\sfS)=\max_{j\in\ell} (|\bd(\sfs_j)|+\val(\sfs_j))$.

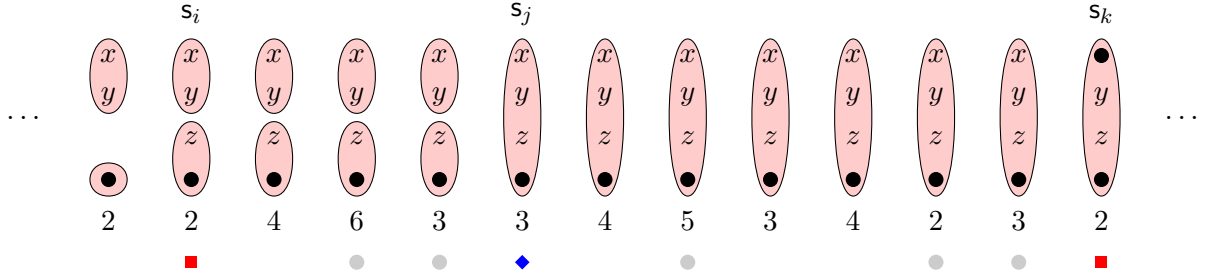
\begin{figure}[h]
\centering
\begin{tikzpicture}[scale=1.1]
   \tikzstyle{vertex}=[fill,circle,minimum size=0.2cm,inner sep=0pt]
   \tikzstyle{type1}=[fill,red,rectangle,minimum size=0.15cm,inner sep=0pt]
   \tikzstyle{type2}=[fill,blue,diamond,minimum size=0.2cm,inner sep=0pt]
   \tikzstyle{type3}=[fill,black!20,circle,minimum size=0.2cm,inner sep=0pt]

\node[] at (-1,-0.75) {\dots};
\node[] at (13,-0.75) {\dots};

%---- 1 ------

\begin{scope}
\coordinate (X) at (0,0);
\coordinate (Y) at (0,-0.5);
\coordinate (Z) at (0,-1);
\coordinate (V) at (0,-2);
\coordinate (T) at (0,-2.5);

\filldraw[fill=red!20] (0,0.2) .. controls (0.3,0.2) and (0.3,-0.7) .. (0,-0.7)  .. controls (-0.3,-0.7) and (-0.3,0.2) .. (0,0.2)
-- cycle;

\filldraw[fill=red!20] (0,-1.3) .. controls (0.3,-1.3) and (0.3,-1.7) .. (0,-1.7)  .. controls (-0.3,-1.7) and (-0.3,-1.3) .. (0,-1.3)
-- cycle;

\node[] at (X) {$x$};
\node[] at (Y) {$y$};
\node[] at (V) {$2$};
%\node[type1] at (T) {};
\node[vertex] at (0,-1.5) {};

\end{scope}

%--- 2------
\begin{scope}[shift={(1,0)}]

\coordinate (X) at (0,0);
\coordinate (Y) at (0,-0.5);
\coordinate (Z) at (0,-1);
\coordinate (V) at (0,-2);
\coordinate (T) at (0,-2.5);

\filldraw[fill=red!20] (0,0.2) .. controls (0.3,0.2) and (0.3,-0.7) .. (0,-0.7)  .. controls (-0.3,-0.7) and (-0.3,0.2) .. (0,0.2)
-- cycle;

\filldraw[fill=red!20] (0,-0.8) .. controls (0.3,-0.8) and (0.3,-1.7) .. (0,-1.7)  .. controls (-0.3,-1.7) and (-0.3,-0.8) .. (0,-0.8)
-- cycle;

\node[] at (X) {$x$};
\node[] at (Y) {$y$};
\node[] at (Z) {$z$};
\node[] at (V) {$2$};
\node[type1] at (T) {};
\node[vertex] at (0,-1.5) {};
\node[] at (0,0.5) {$\sfs_i$};

\end{scope}

%--- 3------
\begin{scope}[shift={(2,0)}]

\coordinate (X) at (0,0);
\coordinate (Y) at (0,-0.5);
\coordinate (Z) at (0,-1);
\coordinate (V) at (0,-2);
\coordinate (T) at (0,-2.5);

\filldraw[fill=red!20] (0,0.2) .. controls (0.3,0.2) and (0.3,-0.7) .. (0,-0.7)  .. controls (-0.3,-0.7) and (-0.3,0.2) .. (0,0.2)
-- cycle;

\filldraw[fill=red!20] (0,-0.8) .. controls (0.3,-0.8) and (0.3,-1.7) .. (0,-1.7)  .. controls (-0.3,-1.7) and (-0.3,-0.8) .. (0,-0.8)
-- cycle;

\node[] at (X) {$x$};
\node[] at (Y) {$y$};
\node[] at (Z) {$z$};
\node[] at (V) {$4$};
\node[vertex] at (0,-1.5) {};

\end{scope}

%--- 4------
\begin{scope}[shift={(3,0)}]

\coordinate (X) at (0,0);
\coordinate (Y) at (0,-0.5);
\coordinate (Z) at (0,-1);
\coordinate (V) at (0,-2);
\coordinate (T) at (0,-2.5);

\filldraw[fill=red!20] (0,0.2) .. controls (0.3,0.2) and (0.3,-0.7) .. (0,-0.7)  .. controls (-0.3,-0.7) and (-0.3,0.2) .. (0,0.2)
-- cycle;

\filldraw[fill=red!20] (0,-0.8) .. controls (0.3,-0.8) and (0.3,-1.7) .. (0,-1.7)  .. controls (-0.3,-1.7) and (-0.3,-0.8) .. (0,-0.8)
-- cycle;

\node[] at (X) {$x$};
\node[] at (Y) {$y$};
\node[] at (Z) {$z$};
\node[] at (V) {$6$};
\node[type3] at (T) {};
\node[vertex] at (0,-1.5) {};

\end{scope}

%--- 5------
\begin{scope}[shift={(4,0)}]

\coordinate (X) at (0,0);
\coordinate (Y) at (0,-0.5);
\coordinate (Z) at (0,-1);
\coordinate (V) at (0,-2);
\coordinate (T) at (0,-2.5);

\filldraw[fill=red!20] (0,0.2) .. controls (0.3,0.2) and (0.3,-0.7) .. (0,-0.7)  .. controls (-0.3,-0.7) and (-0.3,0.2) .. (0,0.2)
-- cycle;

\filldraw[fill=red!20] (0,-0.8) .. controls (0.3,-0.8) and (0.3,-1.7) .. (0,-1.7)  .. controls (-0.3,-1.7) and (-0.3,-0.8) .. (0,-0.8)
-- cycle;

\node[] at (X) {$x$};
\node[] at (Y) {$y$};
\node[] at (Z) {$z$};
\node[] at (V) {$3$};
\node[vertex] at (0,-1.5) {};
\node[type3] at (T) {};

\end{scope}

%--- 6------
\begin{scope}[shift={(5,0)}]

\coordinate (X) at (0,0);
\coordinate (Y) at (0,-0.5);
\coordinate (Z) at (0,-1);
\coordinate (V) at (0,-2);
\coordinate (T) at (0,-2.5);

\filldraw[fill=red!20] (0,0.2) .. controls (0.3,0.2) and (0.3,-1.7) .. (0,-1.7)  .. controls (-0.3,-1.7) and (-0.3,0.2) .. (0,0.2)
-- cycle;

\node[] at (X) {$x$};
\node[] at (Y) {$y$};
\node[] at (Z) {$z$};
\node[] at (V) {$3$};
\node[type2] at (T) {};
\node[vertex] at (0,-1.5) {};
\node[] at (0,0.5) {$\sfs_j$};

\end{scope}

%--- 7------
\begin{scope}[shift={(6,0)}]

\coordinate (X) at (0,0);
\coordinate (Y) at (0,-0.5);
\coordinate (Z) at (0,-1);
\coordinate (V) at (0,-2);
\node[vertex] at (0,-1.5) {};

\filldraw[fill=red!20] (0,0.2) .. controls (0.3,0.2) and (0.3,-1.7) .. (0,-1.7)  .. controls (-0.3,-1.7) and (-0.3,0.2) .. (0,0.2)
-- cycle;

\node[] at (X) {$x$};
\node[] at (Y) {$y$};
\node[] at (Z) {$z$};
\node[] at (V) {$4$};
\node[vertex] at (0,-1.5) {};

\end{scope}

%--- 8------
\begin{scope}[shift={(7,0)}]

\coordinate (X) at (0,0);
\coordinate (Y) at (0,-0.5);
\coordinate (Z) at (0,-1);
\coordinate (V) at (0,-2);
\coordinate (T) at (0,-2.5);

\filldraw[fill=red!20] (0,0.2) .. controls (0.3,0.2) and (0.3,-1.7) .. (0,-1.7)  .. controls (-0.3,-1.7) and (-0.3,0.2) .. (0,0.2)
-- cycle;

\node[] at (X) {$x$};
\node[] at (Y) {$y$};
\node[] at (Z) {$z$};
\node[] at (V) {$5$};
\node[type3] at (T) {};
\node[vertex] at (0,-1.5) {};

\end{scope}

%--- 9------
\begin{scope}[shift={(8,0)}]

\coordinate (X) at (0,0);
\coordinate (Y) at (0,-0.5);
\coordinate (Z) at (0,-1);
\coordinate (V) at (0,-2);
\coordinate (T) at (0,-2.5);

\filldraw[fill=red!20] (0,0.2) .. controls (0.3,0.2) and (0.3,-1.7) .. (0,-1.7)  .. controls (-0.3,-1.7) and (-0.3,0.2) .. (0,0.2)
-- cycle;

\node[] at (X) {$x$};
\node[] at (Y) {$y$};
\node[] at (Z) {$z$};
\node[] at (V) {$3$};
\node[vertex] at (0,-1.5) {};

\end{scope}

%--- 10------
\begin{scope}[shift={(9,0)}]

\coordinate (X) at (0,0);
\coordinate (Y) at (0,-0.5);
\coordinate (Z) at (0,-1);
\coordinate (V) at (0,-2);
\coordinate (T) at (0,-2.5);

\filldraw[fill=red!20] (0,0.2) .. controls (0.3,0.2) and (0.3,-1.7) .. (0,-1.7)  .. controls (-0.3,-1.7) and (-0.3,0.2) .. (0,0.2)
-- cycle;

\node[] at (X) {$x$};
\node[] at (Y) {$y$};
\node[] at (Z) {$z$};
\node[] at (V) {$4$};
\node[vertex] at (0,-1.5) {};

\end{scope}

%--- 11------
\begin{scope}[shift={(10,0)}]

\coordinate (X) at (0,0);
\coordinate (Y) at (0,-0.5);
\coordinate (Z) at (0,-1);
\coordinate (V) at (0,-2);
\coordinate (T) at (0,-2.5);

\filldraw[fill=red!20] (0,0.2) .. controls (0.3,0.2) and (0.3,-1.7) .. (0,-1.7)  .. controls (-0.3,-1.7) and (-0.3,0.2) .. (0,0.2)
-- cycle;

\node[] at (X) {$x$};
\node[] at (Y) {$y$};
\node[] at (Z) {$z$};
\node[] at (V) {$2$};
\node[type3] at (T) {};
\node[vertex] at (0,-1.5) {};

\end{scope}

%--- 12------
\begin{scope}[shift={(11,0)}]

\coordinate (X) at (0,0);
\coordinate (Y) at (0,-0.5);
\coordinate (Z) at (0,-1);
\coordinate (V) at (0,-2);
\coordinate (T) at (0,-2.5);

\filldraw[fill=red!20] (0,0.2) .. controls (0.3,0.2) and (0.3,-1.7) .. (0,-1.7)  .. controls (-0.3,-1.7) and (-0.3,0.2) .. (0,0.2)
-- cycle;

\node[] at (X) {$x$};
\node[] at (Y) {$y$};
\node[] at (Z) {$z$};
\node[] at (V) {$3$};
\node[type3] at (T) {};
\node[vertex] at (0,-1.5) {};

\end{scope}

%--- 13------
\begin{scope}[shift={(12,0)}]

\coordinate (X) at (0,0);
\coordinate (Y) at (0,-0.5);
\coordinate (Z) at (0,-1);
\coordinate (V) at (0,-2);
\coordinate (T) at (0,-2.5);

\filldraw[fill=red!20] (0,0.2) .. controls (0.3,0.2) and (0.3,-1.7) .. (0,-1.7)  .. controls (-0.3,-1.7) and (-0.3,0.2) .. (0,0.2)
-- cycle;

\node[] at (Y) {$y$};
\node[] at (Z) {$z$};
\node[] at (V) {$2$};
\node[type1] at (T) {};
\node[] at (0,0.5) {$\sfs_k$};
\node[vertex] at (0,0) {};
\node[vertex] at (0,-1.5) {};

\end{scope}

\end{tikzpicture}
\caption{The part $\langle \sfs_{i-1},\dots, \sfs_k\rangle$ of a $B$-boundaried sequence $\sfS$ where the boundary set $B$ contains among others the vertices $x$, $y$ and $z$. A bullet $\bullet$ at some index $j$ represents an element of $\bigcup_{h<j}\bd(\sfs_h)$.
Observe that at index $k$, $x$ is indeed represented by a black bullet. For the index $i$, we have $\bd(\sfs_i)=\{x,y,z\}$,
$\cc(\sfs_i)=\{\{x,y\},\{z,\bullet\}\}$ and $\val(\sfs_i)=2$. At every position $j$, only named elements belong to $\bd(\sfs_j)$. The red squares mark the
type-$1$ breakpoints: at position $i$, element $z$ is new, while at position $k$, element $x$ is forgotten. The blue diamond at index $j$ marks a type-$2$
breakpoint which corresponds to the merge of two parts of $\cc(\sfs_{i+4})$ into a single part. Finally, the grey bullets mark type-$3$ breakpoints corresponding to tips of the integer sequences $\langle\val(\sfs_i),\dots, \val(\sfs_{j-1})\rangle$ and $\langle\val(\sfs_j),\dots, \val(\sfs_{k-1})\rangle$. 
\label{fig_boundaried_sequence}}
\end{figure}

\begin{definition}[Connected $B$-boundaried sequence]
Let $\sfS=\langle \sfs_1,\ldots, \sfs_\ell\rangle$ be a $B$-boundaried sequence for some finite set $B$. We say that $\sfS$ is \emph{connected} if for every $i\in[\ell]$, 
$\cc({\sfs_i})$ is a partition of $\bigcup_{i\leq j} \bd(\sfs_i)\subseteq B$.
\end{definition}

Observe that if $\sfS=\langle \sfs_1,\ldots, \sfs_\ell\rangle$ is a connected $B$-boundaried sequence and if there exists some $i\in[\ell]$ such that $\cc(\sfs_i)=\{\emptyset\}$, then, for every $j\leq i$, $\bd(\sfs_j)=\emptyset$ and $\cc(\sfs_j)=\{\emptyset\}$.

As we will see in \autoref{sec_encoding}, the \emph{$B$-boundaried sequences} will allow us to encode partial connected path-decompositions. Intuitively, if $P=\langle A_1,\dots, A_{\ell}\rangle$ is a path-decomposition, a triple $\sfs_j$ will represent the informations about bag $A_j$: $\bd(\sfs_j)$ contains the active vertices of the boundary set $B$; $\val(\sfs_j)$ the number of boundary vertices that appear in prior bags $A_i$ ($i<j$) but not in $A_j$; and $\cc(\sfs_j)$ encodes how the connected components of the graph induced by $\cup_{i\leq j} A_i$ project on $B$.
%------------------------------------------------------------
\subsection{Breakpoints, representatives and domination relation}

\begin{definition}[Breakpoints]
Let $\sfS=\langle \sfs_1,\dots, {\sfs}_j,\dots, \sfs_{\ell}\rangle$ be a $B$-boundaried sequence for some finite set $B$. Then the index $j$, with $1\leq j\leq \ell$, is a \emph{breakpoint} of: 
\begin{itemize}
\item \emph{type-1} if $j=1$  or $\bd(\sfs_j)\neq\bd(\sfs_{j-1})$ or $j=\ell$;
\item \emph{type-2} if it is not a type-1 breakpoint and $\cc(\sfs_j)\neq\cc(\sfs_{j-1})$; 
\item \emph{type-3} if it is not a type-1 nor a type-2 and $j$ is a tip of the integer sequence $\langle \val(\sfs_{l_j}),\dots,$\\$ \val(\sfs_{r_j-1})\rangle$ where $l_j$ and $r_j$ are respectively the largest and smallest type-1 or type-2 breakpoints such that $l_j<j<r_j$. 
\end{itemize}
We denote by $\bp(\sfS)$  the set of breakpoints of $\sfS$ and by $\bp_t(\sfS)$ the set of type-$t$ breakpoints of $\sfS$, for $t\in \{1,2,3\}$.
We define the \emph{representative sequence} $\rep(\sfS)$ of $\sfS$ as the induced subsequence of $\sfS[\bp(\sfS)]$.
\end{definition}

Figure~\ref{fig_boundaried_sequence} illustrates the notions of $B$-boundaried sequence and breakpoints. Observe that
$\rep(\sfS)$ can be computed from the $B$-boundaried sequence
$\sfS$ by an algorithm similar to the one described in Definition~\ref{def:rep-seq} and as in Lemma \ref{lem:size-ts}
$\rep(\sfS)$ is uniquely defined. Notice that, as an induced subsequence of $\sfS$, $\rep(\sfS)$ is a $B$-boundaried sequence. Let
$\ell$ be the length of $\sfS$. It is worth to remark that if $1<j\leq\ell$ belongs to $\bp_1(\sfS)\cup\bp_2(\sfS)$, then
$j-1$ is also a breakpoint. This is the case because the last index of an integer sequence is by definition a tip.

We define the set of representative $B$-boundaried sequences of width at most $\w$ as
$$\Rep_\w(B)=\{\rep(\sfS)\mid \sfS \textrm{ is a $B$-boundaried sequence of width $\leq \w$}\}.$$

\begin{definition}[$B$-boundary model] \label{def:model}
Let $\sfS=\langle {\sfs}_1,\dots, {\sfs}_j,\dots, {\sfs}_{\ell}\rangle$ be a $B$-boundaried sequence. For every $j\in[\ell]$, we set $\dot{\sfs}_j=(\bd(\sfs_j),\cc(\sfs_j),\mathbf{t}(\sfs_j))$ with $\mathbf{t}(\sfs_j)=1$ if $j\in\bp_1(\sfS)$, $\mathbf{t}(\sfs_j)=2$ if $j\in\bp_2(\sfS)$ and $\mathbf{t}(\sfs_j)=0$ otherwise. The \emph{$B$-boundary model} of $\sfS$, denoted by $\model(\sfS)$, is the subsequence of $\dot\sfS=\langle \dot{\sfs}_1,\dots, \dot{\sfs}_j,\dots, \dot{\sfs}_{\ell}\rangle$ induced by $\bp_1(\sfS)\cup\bp_2(\sfS)$. 
\end{definition}

As in \cite{BodlaenderK96effi,JeongKO17the}, we will bound the number of representatives of $B$-boundaried sequences, and for doing so we bound the number of $B$-boundaried models and then use Lemma \ref{lem:nb-ts} which gives an upper bound on the number of typical sequences.

\begin{lemma} \label{lem:model} 
Let $\sfS$ be a $B$-boundaried sequence. If $\sfS^*\in {\sf Ext}(\sfS)$, then $\model(S^*)=\model(S)$.
\end{lemma}
\begin{proof}
This follows from the observation that the duplication of an element of a $B$-boundaried sequence does not generate a new breakpoint nor kill any existant breakpoint.
\end{proof}

\begin{lemma}\label{lem:size-model} 
Let $B$ be a set of size $k$. Then, there are at most {$2k+1$} type-1 breakpoints and at most $k+1$ type-2 breakpoints.
\end{lemma}

\begin{proof} 
Let $\sfS=\langle {\sfs}_1,\dots, {\sfs}_j,\dots, {\sfs}_{\ell}\rangle$ be a $B$-boundaried sequence. By definition, for every $x\in B$, the subset $\{j\in[\ell]\mid x\in\bd(\sfs_j)\}$ forms a set of consecutive integers. So every element $x\in B$ may generate $2$ type-1 breakpoints. This implies that $\sfS$ contains at most $2k+1$ breakpoints.

Let's now consider the number of type-2 breakpoints. By definition of a $B$-boundaried sequence, for every $i<\ell$, we have $\cc(\sfs_i) \sqsubseteq \cc(\sfs_{i+1})$. Moreover if $i,j\in[\ell]$ are two consecutive type-2 breakpoints with $i<j$, then $\cc(\sfs_i)\neq \cc(\sfs_j)$. Observe that if $\cc(\sfs_i)\neq \cc(\sfs_j)$, then either
several blocks of $\cc(\sfs_i)$ are joined into one block in $\cc(\sfs_j)$ or some new block $X$ appears in $\cc(\sfs_j)$ such that $X\cap\bd(\sfs_i)=\emptyset$. Because $|B|=k$ and a near-partition contains at most $k+1$ blocks, by the previous argument we can have at most $k+1$ type-2 breakpoints.
\end{proof}

{
\begin{lemma}\label{lem:nb-models}  
Let $B$ be a set of size $k$. Then, there are $2^{O(k\log k)}$  different $B$-boundary models. 
\end{lemma}

\begin{proof} By Lemma \ref{lem:size-model}, the length of a $B$-boundary model is at most {$3k+2$}.  By definition, each vertex $x\in B$ appears in an interval. Therefore,
  to build a $B$-boundary model, we have to choose, for each vertex $x\in B$, $2$ positions among $3k+2$ ones, therefore there are  $(3k+2)^{2k}=2^{O(k\log k)}$ possibilities for choosing
  the positions of the elements {$\bd(\sfs_j)$} in $B$.  Since  each type-2 breakpoint is assigned a near-partition of at most $k$ blocks on a set of size at most $k$  and these near-partitions are gradually coarsening, the possibilities of assigning them
 correspond to the number of rooted trees on $3k+2$ levels and $k$ leaves.
 As this is bounded by $2^{O(k)}$,  the number of $B$-boundary models is $2^{O(k\log k)}$.
\end{proof}
}

\begin{lemma}\label{lem:nb-rep} 
Let $B$ be a set of size $k$. Then, $|\Rep_\w(B)|={2^{O(k (\w + \log k)}}$. 
\end{lemma}

\begin{proof} 
We only need to bound the number of possible representatives of width $\w$ having the same $B$-boundary model. By Lemma \ref{lem:size-model}, there are at most $3k+2$ type-1 or type-2 breakpoints. Because $\rep(\sfS)$ has size $\bp(\sfS)$ and a type-3 breakpoint is between two type-1 or type-2 breakpoints, we have to bound the number of typical sequences. By Lemma \ref{lem:nb-ts}, the number of typical sequences with integers $\{0,1,\ldots,\w\}$ is at most $\frac{8}{3}\cdot 2^{2\w}{=2^{O(\w)}}$. Since there are at most $3k+2{=O(k)}$ intervals where we can locate type-3 breakpoints, {we have $2^{O(\w k)}$ possible ways to  assign them}.  The lemma now  follows if we take into account  the upper bound by Lemma \ref{lem:nb-models}. 
\end{proof}

Notice that the notion of a $B$-boundary model  corresponds to the one of {\sl interval model} in \cite{BodlaenderK96effi}.
Besides the $B$-boundary model of a sequence $\sfS$, we introduce the \emph{profile} of $\sfS$, which corresponds to the concept of {\sl list representation} in \cite{BodlaenderK96effi}.

\begin{definition}[Profile]
Let $\sfS$ be a $B$-boundaried sequence of length $\ell$ and let $1=j_1< \cdots < j_i< \cdots < j_r=\ell$ be the subset of indices of $[\ell]$ that belong to $\bp_1(\sfS)\cup\bp_2(\sfS)$. Then 
we set $\profile(\sfS)=\langle \mathsf{L}_1, \dots, \mathsf{L}_r\rangle$ with, for $i\in[r]$, $\mathsf{L}_j=\langle \val(\sfs_{j_i}), \dots, \val(\sfs_{j_{i+1}-1})\rangle$. 
\end{definition}

Let us now introduce the domination relation over $B$-boundaried sequences. This relation will allow us to compare $B$-boundaried sequences having the same model  by means of their $B$-profiles. 

\begin{definition}[Domination relation]\label{def:domination}
  Let $\sfS=\langle {\sfs}_1,\dots, {\sfs}_j,\dots, {\sfs}_{\ell}\rangle$ and $\sfT=\langle {\sft}_1,\dots, {\sft}_j,\dots, {\sft}_{\ell}\rangle$ be two $B$-boundaried
  sequences such that $\model(\sfS)=\model(\sfT)$. If $\profile(\sfS)\leq \profile(\sfT)$, then we write $\sfS\leq \sfT$. And, we say that $\sfS$ \emph{dominates} $\sfT$, denoted by
  $\sfS\preceq\sfT$, if $\profile(\sfS)\preceq\profile(\sfT)$. If we have $\profile(\sfS)\preceq\profile(\sfT)$ and $\profile(\sfT)\preceq\profile(\sfS)$, then we say that $\sfS$
  and $\sfT$ are equivalent, which is denoted by $\sfS\equiv\sfT$.
\end{definition}

\begin{lemma} \label{lem:domination-extension}
Let $\sfS$ and $\sfT$ be two $B$-boundaried sequences such that $\model(\sfS)=\model(\sfT)$. If $\sfS\preceq \sfT$, then there exist $\sfS^*$ an extension of $\sfS$ and $\sfT^*$ an extension of $\sfT$ such that
$\sfS^*\leq \sfT^*$.
\end{lemma}
\begin{proof}
This is a direct consequence of the definitions.
\end{proof}

We observe that some properties on integer sequences from Lemma~\ref{lem:prop-useful} transfer to $B$-boundaried sequences, and we state in
the following some of them that we refer to implicitly most of the time (to avoid overloading the text).

\begin{lemma}\label{lem:prop-useful-sequence} Let $\sfS$ be a $B$-boundaried sequence. Then,
\begin{enumerate}
  \item $\rep(\sfS)\equiv \sfS$,
  \item if $\sfS^*\in {\sf Ext}(\sfS)$, then $\sfS^*\equiv \sfS$,
  \item $\sfS \preceq \sfT$ if and only if $\rep(\sfS) \preceq \rep(\sfT)$.
    \item If ${\sf T}$ is a $B$-boundaried sequence such that  $\sfS\preceq \sfT$, then there exist an extension $\sfS^*$ of $\sfS$ and  an extension  $\sfT^*$  of $\sfT$ such that
$\sfS^*\leq \sfT^*$.
\item   The relation $\preceq$ is transitive, and $\equiv$ is an equivalence relation (refering to boundary sequences).
  \end{enumerate}\end{lemma}

\begin{proof} Let's prove (1). By definition $\sfS$ and $\rep(\sfS)$ have the same $B$-boundary model. Let
  $\profile(\sfS)=\langle L_1,\ldots, L_p\rangle$. By definition, $\profile(\rep(\sfS))=\langle \tseq(L_1),\ldots, \tseq(L_p)\rangle$, and by
  Lemma \ref{lem:prop-useful}(6), we know that $\tseq(L_i)\equiv L_i$, for $i\in[p]$. We can therefore conclude that
  $\profile(\sfS) \equiv \profile(\rep(\sfS))$, \ie, $\sfS\equiv \rep(\sfS)$.
  For (2), if $\sfS^*\in {\sf Ext}(\sfS)$, then clearly $\sfS^*\preceq \sfS$ and $\sfS\preceq \sfS^*$ by taking as an extension of $\sfS$ its
  extension $\sfS^*$, and for an extension of $\sfS^*$ itself. Finally, 
  (4) follows directly from the definitions, (5) follows from Lemma \autoref{lem:prop-useful}(5), and (3) follows from (1) and (5).  
\end{proof}

%------------------------------------------------------------
\subsection{Operations on $B$-boundaried sequences}

Given a finite set $B$, we define two operations on $B$-boundaried sequences that will be later used in the DP algorithm. The first operation, \emph{projection}, will be used in the
case of forget bags where we need to transform a $B$-boundaried sequence representing a connected path-decomposition of a boundaried graph $\G=(G,B)$ into a $B\setminus\{x\}$-boundaried sequence representing a connected path-decomposition of the boundaried graph $\Gx=(G,B\setminus\{x\})$. The second operation deals with the insertion in a $B$-boundaried sequence of a new boundary element $x$ with respect to a subset $X\subseteq B$. This will be used by the DP algorithm when handling insertion bags.

%------------------------------------------------------------
\subsubsection{Projection of $B$-boundaried sequences}

The \emph{projection} of a $B$-boundaried sequence $\sfS$ onto $B'\subseteq B$ aims at moving the vertices of $B\setminus B'$ from the status of boundary vertices to the status of inactive vertices.

\begin{definition}[Projection]\label{def:projection}
Let $\sfS=\langle {\sfs}_1,\dots, {\sfs}_i,\dots, {\sfs}_{\ell}\rangle$ be a $B$-boundaried sequence. For a subset $B'\subseteq B$, the \emph{projection} of $\sfS$ onto $B'$ is the $B'$-boundaried sequence $\sfS_{\mid B'}=\langle {\sfs_1}_{\mid B'},\dots, {\sfs_i}_{\mid B'},\dots, {\sfs_{\ell}}_{\mid B'}\rangle$ such that
for every $i\in[\ell]$:
\begin{itemize}
\item $\bd({\sfs_i}_{\mid B'})=\bd(\sfs_i)\cap B'$;
\item $\cc({\sfs_i}_{\mid B'})=\cc(\sfs_i)_{\mid B'}$;
\item $\val({\sfs_i}_{\mid B'})= \val(\sfs_i)+|\bd(\sfs_i)\setminus B'|$.
\end{itemize}
\end{definition}

We observe that when  the $B$-boundaried sequence $\sfS$ is connected, its projection $\sfS_{\mid B'}$ onto $B'\subseteq B$ may not be connected. This is the case if for some $j\in[\ell]$, the partition $\cc(\sfs_j)$ contains several blocks and at least one of them is a subset of $B\setminus B'$.

\begin{lemma}\label{lem:p-width} 
Let $B$ be a finite set and $B'\subseteq B$. Then, the width of $\sfS_{\mid B'}$ is equal to the width of $\sfS$, for every $B$-boundaried sequence $\sfS$.
\end{lemma}

\begin{proof} Let $\sfS=\langle {\sfs}_1,\dots, {\sfs}_j,\dots, {\sfs}_{\ell}\rangle$ and $\sfS_{\mid B'}=\langle {\sfs'}_1,\dots, {\sfs'}_j,\dots, {\sfs'}_{\ell}\rangle$. By
  definition, for each $1\leq j\leq \ell$, $|\bd(\sfs_j)| +\val(\sfs_j)=|\bd(\sfs_j)\cap B'| + |\bd(\sfs_j)\setminus B'| + \val(\sfs_j)$, the latter being exactly
  $|\bd({\sfs'}_j)|+\val({\sfs'}_j)$. 
\end{proof}

\begin{lemma}\label{lem:p-ex-comp} 
Let $B$ be a finite set and $B'\subsetneq B$. If $\sfS^*$ is an extension of a $B$-boundaried sequence $\sfS$, then $\sfS^*_{\mid B'}$ is an extension of $\sfS_{\mid B'}$.
\end{lemma}

\begin{proof} Let $\sfS=\langle \sfs_1,\ldots, \sfs_\ell\rangle$. As by Lemma~\ref{lem:model}, $\model(\sfS)=\model(\sfS^*)$, duplicating $\sfs_i$ and then computing ${\sfs_i}_{\mid B'}$ is the same as computing ${\sfs_i}_{\mid B'}$ and then duplicating the latter. 
\end{proof}

\begin{lemma}\label{lem:p-compatibility1} 
Let $B$ be a finite set and $B'\subseteq B$. If $\sfS$ and $\sfT$ are $B$-boundaried sequences such that  $\sfS \leq \sfT$, then $\sfS_{\mid B'} \leq \sfT_{\mid B'}$.
\end{lemma}

\begin{proof} Let $\sfS=\langle \sfs_1,\ldots, \sfs_\ell\rangle$ and let $\sfT=\langle \sft_1,\ldots, \sft_\ell\rangle$. Because $\model(\sfS)=\model(\sfT)$, we also
  have that $\model(\sfS_{\mid B'})=\model(\sfT_{\mid B'})$. Because $\model(\sfS)=\model(\sfT)$, we can check that $\val({\sfs_i}_{\mid B'})$ and $\val({\sft_i}_{\mid
    B'})$ are both obtained by adding the same value to $\val(\sfs_i)$ and to $\val(\sft_i)$, respectively. Hence, we can conclude that $\sfS_{\mid B'} \leq \sfT_{\mid
    B'}$ because $\profile(\sfS)\leq \profile(\sfT)$. 
\end{proof}

\begin{lemma}\label{lem:p-compatibility2} 
Let $B$ be a finite set and $B'\subseteq B$. If $\sfS$ and $\sfT$ are $B$-boundaried sequences such that  $\sfS \preceq \sfT$, then $\sfS_{\mid B'}\preceq \sfT_{\mid B'}$.
\end{lemma}

\begin{proof} Let $\sfS^*$ and $\sfT^*$ be extensions of $\sfS$ and $\sfT$, respectively, such that $\sfS^*\leq \sfT^*$. By Lemma
  \ref{lem:p-compatibility1}, $\sfS^*_{\mid B'}\leq \sfT^*_{\mid B'}$. By Lemma \ref{lem:p-ex-comp}, $\sfS^*_{\mid B'}$ is an extension of
  $\sfS_{\mid B'}$, \ie, $\sfS_{\mid B'} \equiv \sfS^*_{\mid B'}$ by Lemma \ref{lem:prop-useful-sequence}(2). Similarly, we have
  $\sfT^*_{\mid B'} \equiv \sfT_{\mid B'}$. Hence, we can conclude that $\sfS_{\mid B'} \preceq \sfT_{\mid B'}$.
\end{proof}

%---------------------------------------------------------------------------
\subsubsection{Insertion into a $B$-boundaried sequence}
  
Let $\sfS=\langle \sfs_1,\dots,\sfs_{\ell}\rangle$ be a $B$-boundaried sequence and let $X$ be a subset of $B$. An \emph{insertion position} is a pair of indices $(f_x,l_x)$ such that $1\leq f_x\leq l_x\leq\ell$. An insertion position is \emph{valid with respect to $X$ in $\sfS$} if  $X\subseteq \bigcup_{f_x\leq j\leq l_x} \bd(\sfs_j)$. Let us now formally describe the insertion operation.

\begin{definition}\label{def:sequence-insertion}
Let $\sfS=\langle \sfs_1,\dots,\sfs_{\ell}\rangle$ be a $B$-boundaried sequence and $(f_x,l_x)$ be a valid insertion position with respect to $X\subseteq B$.
Then $\sfS^x=\ins(\sfS,x,X,f_x,l_x)=\langle \sfs^x_1,\dots,\sfs^x_{\ell}\rangle$ is the $(B\cup \{x\})$-boundaried sequence such that for every $j\in[\ell]$:
\begin{itemize}
\item  if $j< f_x$, then $\bd(\sfs^x_j)=\bd(\sfs_j)$; $\cc(\sfs^x_j)=\cc(\sfs_j)$ and $\val(\sfs^x_j)=\val(\sfs_j)$.

\item if $f_x\leq j\leq l_x$, then $\bd(\sfs^x_{j})=\bd(\sfs_{j})\cup\{x\}$; $\cc(\sfs^x_j)$ is obtained by adding a new block $\{x\}$ to $\cc(\sfs_{j})$ and then merging that new block with all the blocks of $\cc(\sfs_{j})$ that contain an element of $X$ (if any); $\val(\sfs^x_{j})=\val(\sfs_{j})$.

\item and otherwise, $\bd(\sfs^x_{j})=\bd(\sfs_{j})$; $\cc(\sfs^x_j)$ is obtained by adding a new block $\{x\}$ to $\cc(\sfs_{j})$ and then merging that new block with all the blocks of $\cc(\sfs_{j})$ that contains an element of $X$ (if any); $\val(\sfs^x_{j})=\val(\sfs_{j})$.
\end{itemize}
\end{definition}

It is worth to notice that a type-2 breakpoint $j$ in a $B$-boundaried sequence $\sfS$ may disappear in $\ins(\sfS,x,X,f_x,l_x)$, because  the insertion of $x$ with respect to $X$ may  merge in $\cc(\sfs^x_{j-1})$ distinct blocks of $\cc(\sfs_{j-1})$ that are joined in $\cc(\sfs_j)$.
 However one can prove that if $j\in\bp_2(\sfS^x)$, then $j\in\bp_2(\sfS)$ (see Figure~\ref{fig:cc-insertion} for an illustration of this property) and if $j\in\bp_3(\sfS^x)$, then $j\in\bp_3(\sfS)$.

\begin{figure}[h]
\centering
\begin{tikzpicture}[scale=1.0]
   \tikzstyle{vertex}=[fill,circle,minimum size=0.2cm,inner sep=0pt]
   \tikzstyle{vertexRoot}=[fill,red,rectangle,minimum size=0.15cm,inner sep=0pt]

%---- cc(s_{j-1}}------

\begin{scope}

\coordinate (A) at (0,0);
\coordinate (B) at (0,-1);
\coordinate (C) at (0,-2);
\coordinate (D) at (0,-3);
\coordinate (E) at (0,-4);
\coordinate (X1) at (0,-4.4);
\coordinate (X2) at (0.9,-4.4);
\coordinate (X3) at (0.9,-1.6);
\coordinate (X4) at (0,-1.6);
\coordinate (X5) at (-1.2,-1.6);
\coordinate (X6) at (-1.2,-4.35);

\filldraw[fill=black!20] (A) circle (0.45) ;
\filldraw[fill=black!20] (B) circle (0.45) ;
\filldraw[fill=black!20] (X1) .. controls (X2) and (X3) .. (X4)  .. controls (X5) and (X6) .. (X1)
-- cycle;

\filldraw[fill=red!20] (A) circle (0.35) ;
\filldraw[fill=red!20] (B) circle (0.35) ;
\filldraw[fill=red!20] (C) circle (0.35) ;
\filldraw[fill=red!20] (D) circle (0.35) ;
\filldraw[fill=red!20] (E) circle (0.35) ;

\coordinate (X) at (-0.5,-3) ;
\node[vertex]  (x) at (X){};
\node[left] at (X) {\scriptsize $x$};
\draw[black,very thick] (X) -- (-0.1,-2.2) ;
\draw[black,very thick] (X) -- (-0.2,-3) ;
\draw[black,very thick] (X) -- (-0.1,-3.8) ;

\node[] at (0,-4.8) {$j-1$};

\node[left] at (-0.4,-1.8) {$C^x$};
\node[] at (A) {\scriptsize $C_5$};
\node[] at (B) {\scriptsize $C_4$};
\node[] at (C) {\scriptsize $C_3$};
\node[] at (D) {\scriptsize $C_2$};
\node[] at (E) {\scriptsize $C_1$};
\end{scope}

%--- cc(s_j)------
\begin{scope}[shift={(3,0)}]

\coordinate (A) at (0,0);
\coordinate (B) at (0,-1);
\coordinate (C) at (0,-2);
\coordinate (D) at (0,-3);
\coordinate (E) at (0,-4);
\coordinate (X1) at (0,-4.4);
\coordinate (X2) at (1,-4.4);
\coordinate (X3) at (1,-0.6);
\coordinate (X4) at (0,-0.6);
\coordinate (X5) at (-1.3,-0.6);
\coordinate (X6) at (-1.3,-4.35);

\filldraw[fill=black!20] (A) circle (0.45) ;
\filldraw[fill=black!20] (X1) .. controls (X2) and (X3) .. (X4)  .. controls (X5) and (X6) .. (X1)
-- cycle;

\filldraw[fill=red!20] (A) circle (0.35) ;
\filldraw[fill=red!20] (0,-0.65) .. controls (-0.5,-0.65) and (-0.5,-2.35) .. (0,-2.35)  .. controls (0.5,-2.35) and (0.5,-0.65) .. (0,-0.65)
-- cycle;
\filldraw[fill=red!20] (D) circle (0.35) ;
\filldraw[fill=red!20] (E) circle (0.35) ;

\coordinate (X) at (-0.5,-3) ;
\node[vertex]  (x) at (X){};
\node[left] at (X) {\scriptsize $x$};
\draw[black,very thick] (X) -- (-0.1,-2.2) ;
\draw[black,very thick] (X) -- (-0.2,-3) ;
\draw[black,very thick] (X) -- (-0.1,-3.8) ;

\node[] at (0,-4.8) {$j$};

%\node[left] at (-0.4,-1.8) {$C^x$};
\node[] at (A) {\scriptsize $C_5$};
\node[] at (B) {\scriptsize $C_4$};
\node[] at (C) {\scriptsize $C_3$};
\node[] at (D) {\scriptsize $C_2$};
\node[] at (E) {\scriptsize $C_1$};

\end{scope}

\end{tikzpicture}
\caption{In red, the partitions $\cc(\sfs_{j-1})=\{C_1,C_2,C_3,C_4,C_5\}$ and $\cc(\sfs_j)=\{C_1,C_2,C_3\cup C_4,C_5\}$ certifying that $j\in \bp_2(\sfS)$. In grey, the partitions $\cc(\sfs^x_{j-1})=\{C^x,C_4,C_5\}$ and $\cc(\sfs^x_j)=\{C^x\cup C_4,C_5\}$ certifying that $j\in \bp_2(\sfS^x)$.
\label{fig:cc-insertion}}
\end{figure}
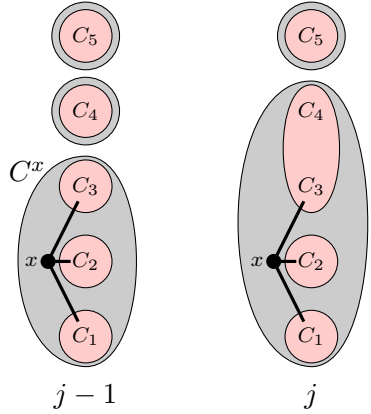

\begin{lemma}\label{lem:i-wd-ins} 
Let $B$ and $B'$ be finite sets with $B=B'\setminus \{x\}$ for some $x\in B'$. Let $\sfS$ be a $B$-boundaried
  sequence and let $(f_x,l_x)$ be a valid insertion position with respect to subset $X\subseteq B$ in $\sfS$. Then, the width of $\sfS$ is at most the width of $\ins(\sfS,x,X,f_x,l_x)$.
\end{lemma}

\begin{proof}
Suppose that $\sfS=\langle \sfs_1,\dots,\sfs_{\ell}\rangle$ and $\ins(\sfS,x,X,f_x,l_x) =\langle \sfs^x_1,\dots,\sfs^x_{\ell}\rangle$. By Definition~\ref{def:sequence-insertion} we have that: for each $1\leq j \leq \ell$, $\val(\sfs^x_j)=\val(\sfs_j)$; if $j\notin  [f_x,l_x]$, then $\bd(\sfs^x_j)=\bd(\sfs_j)$, otherwise  $\bd(\sfs^x_j)=\bd(\sfs_j)\cup\{x\}$.
The statement follows therefore by definition of width of $B$-boundaried sequences. 
\end{proof}

Let us remind that if a $B$-boundaried sequence $\sfT$ of length $p$ is an extension of $\sfS$ of length $\ell$, then the extension surjection $\delta_{\sfT\rightarrow\sfS}:[p]\rightarrow [\ell]$ associates each element of  $\sfT$ with its original copy in $\sfS$ (see Section~\ref{sec:prelim}).

\begin{lemma}\label{lem:i-val-posb}
Let $B$ and $B'$ be finite sets with $B=B'\setminus \{x\}$ for some $x\in B'$. Let $\sfS=\langle \sfs_1,\ldots, \sfs_\ell\rangle$ be a
$B$-boundaried sequence, and let $\sfT\in {\sf Ext}(\sfS)$ that has length $p$ and is  certified by the surjective function
  $\delta_{\sfT\rightarrow\sfS}:[p]\rightarrow [\ell]$.  For every valid insertion position $(f_x,l_x)$ with respect to some subset
$X\subseteq B$ in $\sfS$, $(f^*_x,l^*_x)$ is a valid insertion position  with respect to $X$ in $\sfT$, where $f^*_x=\min\{h\in [p]\mid
f_x=\delta_{\sfT\rightarrow\sfS}(h)\}$ and $l^*_x=\max\{h\in [p]\mid f_x=\delta_{\sfT\rightarrow\sfS}(h)\}$.  Moreover,
$\ins(\sfT,x,X,f^*_x,l^*_x)$ is an extension of $\ins(\sfS,x,X,f_x,l_x)$. 
\end{lemma}
\begin{proof} 
Let us prove the statement for a $1$-extension $\sfT$ of $\sfS$. Inductively applying the proof $p-\ell$ times leads to the statement.

Let us denote $\sfT=\langle \sft_1,\dots,\sft_{\ell+1}\rangle$.  Suppose that $\sfs_i$, for $1\leq i \leq \ell$, is
duplicated, that is for every $j\leq i$, $\delta_{\sfT\rightarrow\sfS}(j)=j$ and for every $i<j\leq \ell+1$, $\delta_{\sfT\rightarrow\sfS}(j)=j-1$.
It is clear that if $i>l_x$ then $(f_x,l_x)$ is still a valid insertion position with respect to $X$ in $\sfT$, and
similarly for $(f_x+1,l_x+1)$ if $i<f_x$. If $f_x\leq i \leq l_x$, then $(f_x,l_x+1)$ is a valid insertion position with respect to $X$ in
$\sfT$ because $\sft_j=\sfs_j$ for $f_x\leq j \leq i$, and $\sfs^*_j= \sfs_{j-1}$ for $i+1\leq j \leq \ell+1$.

We claim now that $\ins(\sfT,x,X,f^*_x,l^*_x)$ is an extension of $\ins(\sfS,x,X,f_x,l_x)$ certified by the surjective
function $\delta_{\sfT\rightarrow\sfS}$. Indeed,  observe that for every $j\in [\ell+1]$, $\sft_j=\sfs_{\delta_{\sfT\rightarrow\sfS}(j)}$.  
So, 
if we duplicate $\sfs^x_i$ in $\sfS^x$, we
will obtain $\ins(\sfT,x,X,f^*_x,l^*_x)$.
\end{proof}

Lemma~\ref{lem:i-val-posb} shows that if $\sfT$ is an extension of $\sfS$, then, to every  valid insertion position $(f_x,l_x)$ with respect to some subset
$X\subseteq B$ in $\sfS$, one can associate a valid insertion position $(f^*_x,l^*_x)$ with respect to $X$  in $\sfT$. As shown by the example of Figure~\ref{fig:counter-example}, the reverse is not true. The following lemma states that  it is indeed possible to associate a valid insertion position $(f^*_x,l^*_x)$ with respect to $X$ in $\sfT$ to some valid insertion position with respect to $X$ in some $(\leq 2)$-extension of $\sfS$.

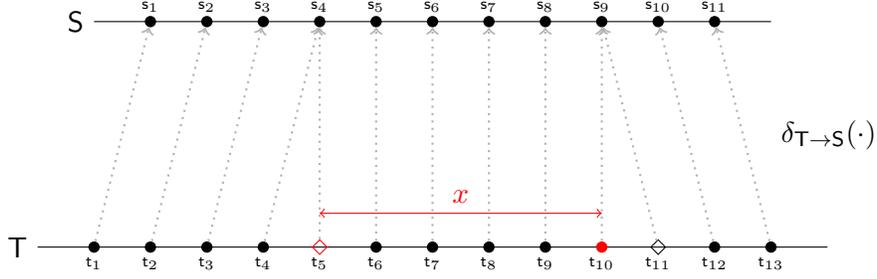
\begin{figure}[h]
\centering
\begin{tikzpicture}[scale=1.5]
   \tikzstyle{vertex}=[fill,black,circle,minimum size=0.15cm,inner sep=0pt]
   \tikzstyle{vertex2}=[draw,black,diamond,minimum size=0.2cm,inner sep=0pt]
   \tikzstyle{vertex3}=[fill,red,circle,minimum size=0.15cm,inner sep=0pt]
   \tikzstyle{vertex4}=[draw,red,diamond,minimum size=0.2cm,inner sep=0pt]

%------ \sfS -------
\draw[black,thin,-] (-3,3) -- (3,3);
\node[anchor=east] at (-3,3) {$\sfS$};

	%\node[vertex] (s0) at (-3,3) {};

\foreach \i in {1,...,11}{
	\node[vertex] (s\i) at (\i/2-3,3) {};
	\node[anchor=south] at (s\i) {\tiny $\sfs_{\i}$};
	}

%------ \sfR -------
\draw[black,thin,-] (-3.5,1) -- (3.5,1);
\node[anchor=east] at (-3.5,1) {$\sfT$};

\foreach \j in {5}{
	\node[vertex4] (r\j) at (\j/2-3.5,1) {};
	\node[anchor=north] at (r\j) {\tiny $\sft_{\j}$};
	}

\foreach \j in {11}{
	\node[vertex2] (r\j) at (\j/2-3.5,1) {};
	\node[anchor=north] at (r\j) {\tiny $\sft_{\j}$};
	}

\foreach \j in {10}{
	\node[vertex3] (r\j) at (\j/2-3.5,1) {};
	\node[anchor=north] at (r\j) {\tiny $\sft_{\j}$};
	}
\foreach \j in {1,2,3,4,6,7,8,9,12,13}{
	\node[vertex] (r\j) at (\j/2-3.5,1) {};
	\node[anchor=north] at (r\j) {\tiny $\sft_{\j}$};
	}

%----- \delta R-S -----
\draw[gray!60,thick,->,dotted] (r1) -- (s1);
\draw[gray!60,thick,->,dotted] (r2) -- (s2);
\draw[gray!60,thick,->,dotted] (r3) -- (s3);
\draw[gray!60,thick,->,dotted] (r4) -- (s4);
\draw[gray!60,thick,->,dotted] (r5) -- (s4);
\draw[gray!60,thick,->,dotted] (r6) -- (s5);
\draw[gray!60,thick,->,dotted] (r7) -- (s6);
\draw[gray!60,thick,->,dotted] (r8) -- (s7);
\draw[gray!60,thick,->,dotted] (r9) -- (s8);
\draw[gray!60,thick,->,dotted] (r10) -- (s9);
\draw[gray!60,thick,->,dotted] (r11) -- (s9);
\draw[gray!60,thick,->,dotted] (r12) -- (s10);
\draw[gray!60,thick,->,dotted] (r13) -- (s11);

\node[anchor=west] at (3,2) {$\delta_{\sfT\rightarrow \sfS}(\cdot)$};

\draw[red,thin,<->] (-1,1.3) -- (1.5,1.3);
\node[red,anchor=south] at (0.25,1.3) {$x$};

\end{tikzpicture}
\caption{Let $\sfT$ be a $2$-extension of the $B$-boundaried sequence $\sfS$. Suppose that $(5,10)$ is a valid insertion position with respect to some for
  $X\subseteq B$ in $\sfT$. Observe that  as $4=\delta_{\sfT\rightarrow \sfS}(5)$ and $9=\delta_{\sfT\rightarrow \sfS}(10)$, $(4,9)$ is also a  valid insertion
  position with respect to some for $X\subseteq B$ in $\sfS$. However,  $\ins(\sfT,x,X,5,10)$ is not an extension of $\ins(\sfS,x,X,4,9)$.
\label{fig:counter-example}}
\end{figure}

\begin{lemma}\label{lem:i-ext} Let $B$ and $B'$ be finite sets with $B=B'\setminus \{x\}$ for some $x\in B'$. Let $\sfT$ be an extension of  
 a $B$-boundaried sequence $\sfS$. If $(f^*_x,l^*_x)$ is a valid insertion position with respect to a subset $X\subseteq B$ in $\sfT$, then there is a $(\leq 2)$-extension $\sfR$ of
 $\sfS$ and a valid insertion position $(f_x,l_x)$ with respect to $X$ in $\sfR$ such that $\ins(\sfT,x,X,f^*_x,l^*_x)$  is an extension of $\ins(\sfR,x,X,f_x,l_x)$.
\end{lemma}

\begin{proof} 
Suppose that $\sfS=\langle \sfs_1,\ldots, \sfs_\ell\rangle$ and $\sfT=\langle \sft_1,\ldots, \sft_p\rangle$. Let $\delta_{\sfT\rightarrow \sfS}:[p]\rightarrow [\ell]$ be the surjection certifying that $\sfT\in {\sf Ext}(\sfS)$, that is for every $j\in[p]$, if $\delta_{\sfT\rightarrow \sfS}(j)=i$, then $\sft_j$ is a copy originating from $\sfs_i$. 
Let us denote $f=\delta_{\sfT\rightarrow \sfS}(f^*_x)$ and $l=\delta_{\sfT\rightarrow \sfS}(l^*_x)$. We also define $f'_x=\min \{j\in[p] \mid \delta_{\sfT\rightarrow \sfS}(j)=f \}$
 and $l'_x=\max \{j\in[p] \mid \delta_{\sfT\rightarrow \sfS}(j)=l \}$. The $(\leq 2)$-extension $\sfR$ of $\sfS$ is built as follows: if $f'_x<f^*_x$, then we duplicate $\sfs_f$ and if $l^*_x<l'_x$, then we duplicate $\sfs_l$. Let $r$ be the size of $\sfR$ and let $\delta_{\sfR\rightarrow \sfS}:[r]\rightarrow[\ell]$ certifying that $\sfR$ is a $(\leq 2)$-extension of $\sfS$.

\begin{figure}[h]
\centering
\begin{tikzpicture}[scale=1.14]
   \tikzstyle{vertex}=[fill,red,circle,minimum size=0.15cm,inner sep=0pt]
   \tikzstyle{vertex2}=[fill,black,diamond,minimum size=0.15cm,inner sep=0pt]

%------ \sfS -------
\draw[black,dotted,-] (-3,3) -- (-2,3);
\draw[black,thin,-] (-2,3) -- (2,3);
\draw[black,dotted,->] (2,3) -- (3,3);
\node[anchor=east] at (-3,3) {$\sfS$};

\coordinate (F) at (-1,3);
\coordinate (L) at (1,3);

\node[vertex] (f) at (F) {};
\node[vertex] (l) at (L) {};

\node[anchor=south] at (F) {\footnotesize $f$};
\node[anchor=south] at (L) {\footnotesize $l$};

%------ \sfR -------
\draw[black,dotted,-] (-6.7,2) -- (-6,2);
\draw[black,thin,-] (-6,2) -- (-1,2);
\draw[black,dotted,->] (-1,2) -- (-0.3,2);
\node[anchor=east] at (-6.7,2) {$\sfR$};

\coordinate (FR') at (-5,2);
\coordinate (FRX) at (-4.5,2);
\coordinate (LRX) at (-2.5,2);
\coordinate (LR') at (-2,2);

\node[vertex2] (f') at (FR'){};
\node[vertex] (fx) at (FRX){};
\node[vertex] (lx) at (LRX){};
\node[vertex2] (l') at (LR'){};

\node[anchor=north] at (FR') {\footnotesize $f'$};
\node[anchor=north] at (FRX) {\footnotesize $f_x$};
\node[anchor=north] at (LRX) {\footnotesize $l_x$};
\node[anchor=north] at (LR') {\footnotesize $l'$};

%------ \sfT -------
\draw[black,dotted,-] (0.3,0.5) -- (1,0.5);
\draw[black,thin,-] (1,0.5) -- (6,0.5);
\draw[black,dotted,->] (6,0.5) -- (6.7,0.5);
\node[anchor=west] at (6.7,0.5) {$\sfT$};

\coordinate (FT') at (1.5,0.5);
\coordinate (FTX) at (2.5,0.5);
\coordinate (LTX) at (4.5,0.5);
\coordinate (LT') at (5,0.5);

\node[vertex2] (ft') at (FT'){};
\node[vertex] (f'x) at (FTX){};
\node[vertex] (ltx) at (LTX){};
\node[vertex2] (lt') at (LT'){};

\node[anchor=north] at (FT') {\footnotesize $f'_x$};
\node[anchor=north] at (FTX) {\footnotesize $f^*_x$};
\node[anchor=north] at (LTX) {\footnotesize $l^*_x$};
\node[anchor=north] at (LT') {\footnotesize $l'_x$};
%----- \delta T-S ------
\draw[gray!60,thick,dotted] (FT') -- (F);
\draw[gray!60,thick,dotted] (FTX) -- (F);
\draw[gray!60,thick,dotted] (LTX) -- (L);
\draw[gray!60,thick,dotted] (LT') -- (L);

\node[anchor=west] at (3.2,1.75) {\scriptsize $\delta_{\sfT\rightarrow\sfS}(\cdot)$};

%----- \delta R-S -----
\draw[gray!60,thick,dotted] (FR') -- (F);
\draw[gray!60,thick,dotted] (FRX) -- (F);
\draw[gray!60,thick,dotted] (LRX) -- (L);
\draw[gray!60,thick,dotted] (LR') -- (L);

\node[anchor=east] at (-3.7,2.5) {\scriptsize $\delta_{\sfR\rightarrow\sfS}(\cdot)$};

%----- \delta T-R -----
\draw[gray!60,thick,dotted] (FR') -- (FT');
\draw[gray!60,thick,dotted] (FRX) -- (FTX);
\draw[gray!60,thick,dotted] (LRX) -- (LTX);
\draw[gray!60,thick,dotted] (LR') -- (LT');

\node[anchor=east] at (-0.6,0.8) {\scriptsize $\delta_{\sfT\rightarrow\sfR}(\cdot)$};

%---- X,x -----

\draw[red,thick,<->] (-4.5,1.5) -- (-2.5,1.5);
\draw[red,thick,<->] (2.5,0) -- (4.5,0);

\node[anchor=north,red] at (-3.5,1.5) {\scriptsize $x$};
\node[anchor=north,red] at (3.5,0) {\scriptsize $x$};

\end{tikzpicture}
\caption{The three surjective functions $\delta_{\sfT\rightarrow\sfS}(\cdot)$, $\delta_{\sfR\rightarrow\sfS}(\cdot)$ and $\delta_{\sfT\rightarrow\sfR}(\cdot)$ respectively certifying that $\sfT\in {\sf Ext}(\sfS)$, $\sfR\in {\sf Ext}(\sfS)$ and $\sfT\in {\sf Ext}(\sfR)$ in the case $f'_x\neq f^*_x$ and $l'_x\neq l^*_x$. In this case,  as $\sfR$ is a $2$-extension of $\sfS$, $f'=\min\{h\in[r]\mid \delta_{\sfR\rightarrow \sfS}(h)=f \}$ 
and $l'=\max\{h\in[r]\mid \delta_{\sfR\rightarrow \sfS}(h)=l \}$.
\label{fig:2-extension}}
\end{figure}
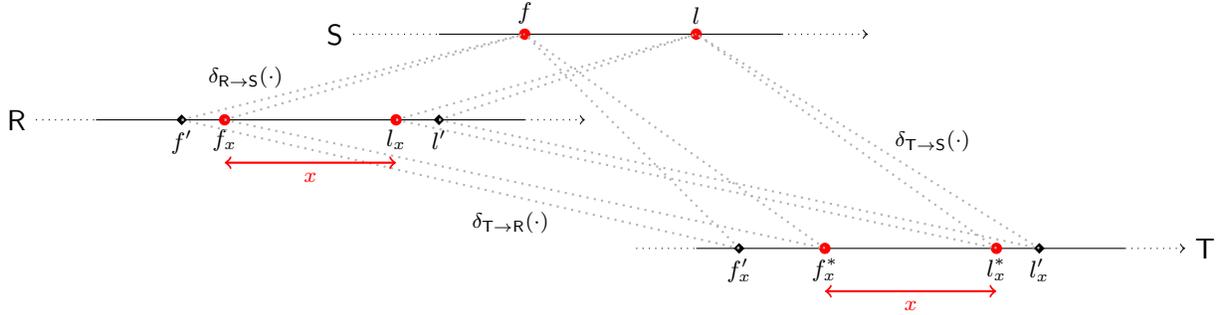

Let us build a surjection $\delta_{\sfT\rightarrow \sfR}:[p]\rightarrow [r]$ certifying that $\sfT$ is an extension of $\sfR$. To that aim, we define $f_x=\max\{h\in[r]\mid \delta_{\sfR\rightarrow \sfS}(h)=f \}$ and $l_x=\min\{h\in[r]\mid \delta_{\sfR\rightarrow \sfS}(h)=l \}$. Then:
\begin{align*}
\delta_{\sfT\rightarrow\sfR}(j) &= 
\begin{cases} 
\delta_{\sfT\rightarrow \sfS}(j) & \textrm{if $j<f_x^*$},\\
\delta_{\sfT\rightarrow \sfS}(j)-f+f_x & \textrm{if $f^*_x\leq j\leq l_x^*$},\\
\delta_{\sfT\rightarrow \sfS}(j)-l+l' & \textrm{if $l^*_x< j$}.
\end{cases}
\end{align*}
where as in Figure~\ref{fig:2-extension} $l'=\max\{h\in[r]\mid \delta_{\sfR\rightarrow \sfS}(h)=l \}$.

Observe that as $\sfT\in {\sf Ext}(\sfS)$ and $\sfR\in {\sf Ext}(\sfS)$, by Lemma~\ref{lem:model}, we have $\model(\sfR)=\model(\sfT)$. Thereby 
$\delta_{\sfT\rightarrow \sfS}(f^*_x)=f_x$ and $\delta_{\sfT\rightarrow \sfS}(l^*_x)=l_x$ implies that $(f_x,l_x)$ is a valid insertion
position with respect to $X$ in $\sfR$. It remains to prove that $\sfT^x=\ins(\sfT,x,X,f^*_x,l^*_x)$  is an extension of
$\sfR^x=\ins(\sfR,x,X,f_x,l_x)$. Observe that, by construction of $\sfR$, $f^*_x=\min\{j\in[p]\mid \delta_{\sfT\rightarrow \sfR}(j)=f_x\}$
and $l^*_x=\max\{j\in[p]\mid \delta_{\sfT\rightarrow \sfR}(j)=l_x\}$. This implies that we can certify $\sfT^x\in {\sf Ext}(\sfR^x)$ by
Lemma~\ref{lem:i-val-posb}. 
\end{proof}

\begin{lemma}\label{lem:i-ext2} 
Let $B$ and $B'$ be finite sets with $B=B'\setminus \{x\}$ for some $x\in B'$. Let $\sfS$ and $\sfT$ be $B$-boundaried sequences such that $\sfS\leq \sfT$. If $(f_x,l_x)$ is a valid insertion position  with respect to a subset $X\subseteq B$ in $\sfT$, then $(f_x,l_x)$ is a valid insertion position with respect to $X$ in $\sfS$ and $\ins(\sfS,x,X,f_x,l_x)\leq \ins(\sfT,x,X,f_x,l_x)$.
\end{lemma}

\begin{proof}
Suppose that $\profile(\sfS)=\langle L_1,\ldots,L_r\rangle$ and $\profile(\sfT)=\langle L'_1,\ldots, L'_r\rangle$. By Definition~\ref{def:domination}, as
$\sfS\leq \sfT$, $\sfT$ and $\sfS$ have the same $B$-model. 
It follows that $(f_x,l_x)$  is a valid insertion position with respect to $X$ in $\sfS$ as well.
And it implies that for every $i\in[r]$, $i\in\bp_1(\sfS)$ if and only if $i\in\bp_1(\sfT)$ and that $i\in\bp_2(\sfS)$ if and only if
$i\in\bp_2(\sfT)$. Thereby, if we denote $\sfS^x=\ins(\sfS,x,X,f_x,l_x)$ and $\sfT^x=\ins(\sfT,x,X,f_x,l_x)$, by Definition~\ref{def:sequence-insertion}, we
obtain that, for every $i\in[r]$, $i\in\bp_1(\sfS^x)$ if and only if $i\in\bp_1(\sfT^x)$ and that $i\in\bp_2(\sfS^x)$ if and only if $i\in\bp_2(\sfT^x)$. Thereby we have $\model(\sfS^x)=\model(\sfT^x)$.
Observe moreover that $\sfS\leq \sfT$ implies that for every $i\in[r]$, $\val(\sfs_i)\leq \val(\sft_i)$. As for every $i\in[r]$, we have that $\val(\sfs_i)=\val(\sfs^x_i)$ and $\val(\sft_i)=\val(\sft^x_i)$, we obtain that $\val(\sfs^x_i)\leq\val(\sft^x_i)$. It follows that $\profile(\sfS^x)\leq\profile(\sfT^x)$, in other words $\sfS^x\leq \sfT^x$.
\end{proof}

\begin{lemma}\label{lem:i-completeness2} 
Let $B$ and $B'$ be finite sets with $B=B'\setminus \{x\}$ for some $x\in B'$. Let $\sfS$ and $\sfT$ be $B$-boundaried sequences such that $\sfS\preceq \sfT$. If $(f_x,l_x)$ is a valid insertion position  with respect to a subset $X\subseteq B$ in $\sfT$, then there is a valid insertion position $(f'_x,l'_x)$ in a $(\leq 2)$-extension $\sfR$ of $\sfS$ such that $\ins(\sfR,x,X,f'_x,l'_x) \preceq \ins(\sfT,x,X,f_x,l_x)$. 
\end{lemma}

\begin{proof} 
Let $\sfS^*$ and $\sfT^*$ be extensions of $\sfS$ and $\sfT$, respectively, such that $\sfS^*\leq \sfT^*$. Suppose that $\sfT^*$ has size $p^*$. Let $\delta_{\sfT^*\rightarrow\sfT}$ be the surjective function certifying that $\sfT^*\in {\sf Ext}(\sfT)$. Let us denote $f^*_x=\min\{h\in [p^*]\mid f_x=\delta_{\sfT^*\rightarrow \sfT}(h)\}$ and $l^*_x=\max\{h\in [p^*]\mid l_x=\delta_{\sfT^*\rightarrow \sfT}(h)\}$. As $(f_x,l_x)$ is a valid insertion position with respect to $X$ in $\sfT$, then by Lemma~\ref{lem:i-val-posb}, $(f_x^*,l_x^*)$ is also a valid insertion position  with respect to $X$ in $\sfT^*$ and  $\ins(\sfT^*,x,X,f^*_x,l^*_x)$ is an extension of $\ins(\sfT,x,X,f_x,l_x)$. 

 By Lemma \ref{lem:i-ext2}, $(f^*_x,l^*_x)$ is a valid insertion position with respect to $X$ in $\sfS^*$ and by Lemma \ref{lem:i-ext} there is a
  $(\leq 2)$-extension $\sfR$ of $\sfS$ and a valid insertion position $(f'_x,l'_x)$ in $\sfR$ such that $\ins(\sfS^*,x,X,f^*_x,l^*_x)$ is an extension of
  $\ins(\sfR,x,X,f'_x,l'_x)$. Then, by Lemma \ref{lem:i-ext2}, we have 
$$\ins(\sfS^*,x,X,f^*_x,l^*_x) \leq \ins(\sfT^*,x,X,f_x^*,l_x^*).$$ 
By using  Lemma~\ref{lem:prop-useful-sequence}(2), it follows    that                                                                    
$ \ins(\sfT,x,X,f_x,l_x) \equiv \ins(\sfT^*,x,X,f^*_x,l^*_x)$
and
 $   \ins(\sfR,x,X,f'_x,l'_x) \equiv  \ins(\sfS^*,x,X,f^*_x,l^*_x),$
  implying the statement by Lemma \ref{lem:prop-useful-sequence}(5).
\end{proof}

%------------------------------------------------------------------------------------------------------------------------
\section{Computing the connected pathwidth}

\label{algosec}

We first explain how $B$-boundaried sequence are natural combinatorial objects to encode a connected path-decomposition. We describe and analyze the time complexity of the \emph{Forget Routine} and the \emph{Insertion Routine} that allow us to respectively process forget and insertion bags of the nice path-decomposition given as input to the DP algorithm.

\subsection{Encoding a connected path-decomposition}
\label{sec_encoding}

Let us explain how to represent a path-decomposition of a boundaried graph $(G,B)$  by means of a $B$-boundaried sequence.

\begin{definition}[$(\mathbf{G},\sfP)$-encoding sequence]
\label{def:canonical-sequence}
Let $\sfP=\langle A_1,\dots, A_\ell\rangle$ be a path-decomposition of the boundaried graph $\mathbf{G}=(G,B).$ 
A $B$-boundaried sequence $\sfS=\langle \sfs_1,\dots, {\sfs}_j,\dots, \sfs_{\ell}\rangle$ is a \emph{$(\G,\sfP)$-encoding sequence}, if for every $j\in[\ell]$:

\begin{itemize}
\item $\bd(\sfs_j)=A_{j}\cap B$: the set of boundary  vertices of $(G,B)$ belonging to the bag $A_{j}$;
\item $\cc(\sfs_j)=\{V(C)\cap B\mid C \mbox{ is a connected component of } G_{j}\}$;
\item $\val(\sfs_j)=|A_{j}\setminus B|$: the number of inactive vertices in the bag $A_{j}$.
\end{itemize}
\end{definition}

It is worth to observe that $\cc(\sfs_j)$ is, in general, not a partition of $A_{j}$ (see Figure~\ref{fig_boundaried_sequence}).  Also, notice that if $G_{j}$ is connected and $B\cap V_{j}=\emptyset,$ then $\cc(\sfs_j)=\{\emptyset\}.$ 

\begin{lemma} \label{lem:connected-seq}
Let $\sfP$ be a path-decomposition of a connected boundaried graph $\mathbf{G}=(G,B).$ If $\sfP$ is a connected path-decomposition, then its $(\G,\sfP)$-encoding sequence is a connected $B$-boundaried sequence.
\end{lemma}
\begin{proof}
Follows directly from the definitions.
\end{proof}

\begin{definition}
Let $\G=(G,B)$ be a {connected} boundaried graph and $\sfS$ a $B$-boundaried sequence. We say that $\sfS$ is \emph{realizable in $\G$} if  there is an extension $\sfS^*$ of $\sfS$ that is the $(\G,\sfP)$-encoding sequence of some connected path-decomposition $\sfP$ of $\G.$
\end{definition}

Let us observe that if a $B$-boundaried sequence $\sfS$ is realizable, then by Lemma~\ref{lem:connected-seq} $\sfS$ is connected.
The set of \emph{representative $B$-boundaried sequences of a connected boundaried graph $\G=(G,B)$ of width $\leq \w$} is defined as:
$$\Rep_\w(\G)=\{\rep(\sfS)\mid \sfS  \textrm{ of width $\leq w$  is realizable in $\G=(G,B)$}\}.$$

To compute the connected pathwidth of a graph, rather than computing $\Rep_\w(\G),$ we compute a subset $\cD_\w(\G)\subseteq\Rep_\w(\G),$ called \emph{domination set}, such that for every representative $B$-boundaried sequence $\sfS\in\Rep_\w(\G),$ there exists a representative $B$-boundaried sequence $\sfR\in\cD_\w(\G)$ such that $\sfR\preceq\sfS.$

\begin{proposition}\label{prop:correctness}
A connected boundaried graph $\G=(G,B)$ has connected pathwidth at most $w$ if and only if $\cD_{\w+1}(\G)\neq\emptyset.$
\end{proposition}

\begin{proof} 
  Let $\sfP$ be a connected path-decomposition of width at most $\w$ of $\G.$ Recall the the bags of such decomposition have size at most
  $\w+1.$ By definition, the $(\G,\sfP)$-encoding sequence is realizable in $\G,$ implying that $\Rep_{\w+1}(\G)$ and thereby
  $\cD_{\w+1}(\G)$ is not empty. Conversely, suppose that $\Rep_{\w+1}(\G)$ is non-empty and consider $\sfS\in\cD_{\w+1}(\G).$ As
  $\sfS\in\Rep_{\w+1}(\G),$ there exists a connected path-decomposition $\sfP$ of width at most $\w$ of $\G$ and $\sfS^*$ the
  $(\G,\sfP)$-encoding sequence with $\rep(\sfS^*)=S$, implying that $\cpw(\G)\leq w.$
\end{proof}

%-----------------------
\subsection{Forget Routine}

Let $\G=(G,B)$ be a boundaried graph. If $x\in B$ is a boundary vertex, we denote by $\Bx=B\setminus\{x\}.$ We define $\Gx=(G,\Bx),$ that is, while the graph $G$ is left unchanged, we remove $x$ from the set of boundary vertices. Given $\cD_\w(\G)$ and $x\in B,$  \emph{Forget Routine} aims at computing a domination set $\cD_\w(\Gx).$ The routine is described in Algorithm~\ref{alg:forget}.

\begin{algorithm}[h]
{\small 
\KwIn{A boundaried graph $\G=(G,B),$ a vertex $x\in B,$ and $\cD_\w(\G).$}
\KwOut{$\cD_\w(\Gx)$, a domination set of $\Rep_\w(\Gx).$}
\BlankLine
$\cD_\w(\Gx)\leftarrow \emptyset$\;
\ForEach{$\sfS\in \cD_\w(\G)$} {
	\lIf{$\sfS_{\mid B\setminus\{x\}}$ is connected,}{
		add $\rep(\sfS_{\mid B\setminus\{x\}})$ to $\cD_\w(\Gx)$
		}
	}
	\Return{ $\cD_\w(\Gx).$}
\caption{Forget Routine} \label{alg:forget}
}
\end{algorithm}

To prove the correctness of  Forget Routine, we proceed in two steps. We first establish the \emph{completeness} of the algorithm. More precisely, Proposition~\ref{prop:forget-complete} states that, for every connected path-decomposition $\sfP$ of $\Gx,$ there exists some $B$-boundaried sequence $\sfS\in\cD_\w(\G)$ such that $\rep(\sfS_{\mid B\setminus\{x\}})\preceq \rep(\sfT)$ where $\sfT$ is the $(\Gx,\sfP)$-encoding sequence. 
Then Proposition~\ref{prop:forget-sound} proves the soundness of the routine: for every $B$-boundaried sequence $\sfS\in \cD_\w(\G),$
$\rep(\sfS_{\mid B\setminus\{x\}})\in\cD_\w(\Gx)$ if $\sfS_{\mid B\setminus\{x\}}$ is connected.

\begin{proposition}[Forget completeness] \label{prop:forget-complete}
Let $\G=(G,B)$ be a boundaried graph and $x\in B$ be a boundary vertex. If $\sfP$ is a connected path-decomposition of width at most $\w$ of $\Gx,$ then there exists $\sfS\in \cD_\w(\G)$ such that $\sfS_{\mid \Bx}$ is connected and $\rep(\sfS_{\mid \Bx})\preceq \rep(\sfT)$ where $\sfT$ is the $(\Gx,\sfP)$-encoding sequence.
\end{proposition}
\begin{proof}
Suppose that $\sfP=\langle A_1,\dots, A_\ell\rangle.$ Observe that $\sfP$ is also a connected path-decomposition of $\G$ of width at most $\w.$ Let $\sfR=\langle \sfr_1,\dots, \sfr_\ell\rangle$ be the $(\G,\sfP)$-encoding sequence. 

We claim that $\sfR_{\mid \Bx}$ is the $(\Gx,P)$-encoding sequence. 
To see this, we apply Definition~\ref{def:projection} on the projection of $\sfR$ onto $\Bx.$ Consider an index $j\in[\ell].$ 
First, we have that $\bd({\sfr_j}_{\mid \Bx})=\bd(\sfr_j)\cap \Bx.$ As by construction of $\sfR,$ $\bd(\sfr_j)=A_j\cap \Bx$ and as $\Bx\subset B,$ we obtain $\bd({\sfr_j}_{\mid \Bx})=A_j\cap \Bx.$ For the same arguments, observe that $\val({\sfr_j}_{\mid \Bx})=\val(\sfr_j)+|\bd(\sfr_j)\setminus \Bx|=|A_j\setminus \Bx|.$ 
Let us now examine $\cc({\sfr_j}_{\mid \Bx})=\cc(\sfr_j)_{\mid \Bx}.$ By Definition~\ref{def:projection}, every block $X\in \cc({\sfr_j}_{\mid \Bx})$ is obtained as $X=X'\cap \Bx$ for some block $X'$ of $\cc(\sfr_j).$ Since $\sfR$ is connected, $X'=C\cap B$ for some connected component $C$ of $G_j=G[V_j],$ and thereby $X=C\cap\Bx.$ The assumption that $\Gx$ is connected implies that if $X=\emptyset,$ then $G_j$ is connected (that is $C=V_j$) and $\Bx\cap V_j=\emptyset$ (that is $B=\{x\}$). This implies that $\cc({\sfr_j}_{\mid \Bx})$ is a partition and fulfills the requirements of Definition~\ref{def:canonical-sequence}.
It follows that $\sfR_{\mid \Bx}$ is indeed the $(\Gx,\sfP)$-encoding sequence and we can thereby set $\sfT=\sfR_{\mid \Bx}.$

Since $\cD_\w(\G)$ is a domination set of $\Rep_\w(\G),$ there exists a $B$-boundaried sequence $\sfS\in \cD_\w(\G)$ such that $\sfS\preceq
\rep(\sfR)$. As $\model(\sfR)=\model(\sfS),$ by Lemma \ref{lem:p-compatibility2} we can conclude that $\sfS_{\mid \Bx}\preceq \sfR_{\mid
  \Bx}=\sfT.$ Lemma~\ref{lem:prop-useful-sequence}(3) allows to conclude that  $\rep(\sfS_{\mid \Bx})\preceq \rep(\sfT).$
\end{proof}

\begin{proposition}[Forget soundness] \label{prop:forget-sound}
Let $\G=(G,B)$ be a boundaried graph and $x\in B$ be a boundary vertex. If $\sfS\in \cD_\w(\G)$ and $\sfS_{\mid\Bx}$ is connected, then $\rep(\sfS_{\mid \Bx})\in\Rep_\w(\Gx).$
\end{proposition}
\begin{proof}
  As $\sfS\in\cD_\w(\G)\subseteq \Rep_\w(\G),$ there exists a connected path-decomposition $\sfP$ of $\G$ of width at most $\w$ such that
  the $(\G,\sfP)$-encoding sequence $\sfT=\langle\sft_1,\dots, \sft_p\rangle$ satisfies $\sfS=\rep(\sfT).$ Since
  $\model(\sfS)=\model(\sfT),$ the hypothesis that $\sfS_{\mid\Bx}$ is connected implies that $\sfT_{\mid\Bx}$ is also connected. It follows
  that $\sfP$ is also a connected path-decomposition of $\Gx.$ One can check that $\sfT_{\mid\Bx}$ is the $(\Gx,\sfP)$-encoding sequence
  (for this, one may just copy the corresponding argument of Proposition \autoref{prop:forget-complete}). As $\sfS=\rep(T),$ we have that
  $\sfS\equiv\sfT$ by Lemma \ref{lem:prop-useful-sequence}(1) and then $\model(\sfS)=\model(\sfT).$ Then, Lemma~\ref{lem:p-compatibility2}
  implies that $\sfS_{\mid \Bx}\equiv\sfT_{\mid \Bx}$ and so $\rep(\sfS_{\mid \Bx})=\rep(\sfT_{\mid \Bx})$ by Lemma
  \ref{lem:prop-useful-sequence} and the fact that the representative is uniquely defined. Finally, as $\sfS$ has width at most $\w$ (it
  belongs to $\cD_\w(\G)$), by Lemma~\ref{lem:p-width}, $\sfS_{\mid \Bx}$ has width at most $\w$ as well. It follows that
  $\rep(\sfS_{\mid \Bx})\in\Rep_\w(\Gx).$
\end{proof}

\begin{theorem} \label{th:forget}
Algorithm~\ref{alg:forget} computes $\cD_\w(\G^{\overline{x}})$ in $2^{O(k (w + \log k))}$-time, where $k=|B|.$
\end{theorem}
\begin{proof}
  The correctness of Algorithm~\ref{alg:forget} is proved by Proposition~\ref{prop:forget-complete} and
  Proposition~\ref{prop:forget-sound}. These two propositions imply that by applying  Forget Routine on a domination set of ${\bf G}$
  included in the set of representatives of $\G$, we indeed compute a domination set of $\Gx$ that is a subset of the set of representatives
  of $\Gx$.  As performing the projection of $B$-boundaried sequence onto $\Bx$ can be performed in polynomial time in the size of the
  sequence, the complexity of the algorithm is dominated by the size of $\cD_\w(\G)$ that is $2^{O(k (w + \log k))}$, because of
  Lemma~\autoref{lem:nb-rep}. 
\end{proof}

%-----------------------
\subsection{  Insertion Routine}

In this subsection, we present the \emph{Insertion Routine}. Suppose that $\G=(G,B)$ is a boundaried graph with $G=(V,E).$ For a subset $X\subseteq B,$ we set $G^{x}=(V\cup\{x\},E\cup\{xy\mid y\in X\})$ and $\G^{x}=(G^{x},B^{x})$ where $B^{x}=B\cup\{x\}.$ Given a domination set $\cD_\w(\G)$ of $\Rep_\w(\G),$ the task of  Insertion Routine is to compute a domination set $\cD_\w(\G^{x})$ of $\Rep_\w(\G^{x}).$ Algorithm~\ref{alg:insertion} is describing  Insertion Routine.

\begin{algorithm} [h]
{\small 
\KwIn{A boundaried graph $\G=(G,B),$ a subset $X\subset B,$ and $\cD_\w(\G).$}
\KwOut{$\cD_\w(\G^x)$, a domination set of $\Rep_\w(\G^x).$}

\BlankLine
$\cD_\w(\G^x)\leftarrow \emptyset$\;
\ForEach{$\sfS=\langle \sfs_1,\ldots, \sfs_\ell\rangle\in \cD_\w(\G)$} {
	\ForEach{$f,l\in[\ell]$ such that $X\subseteq \bigcup_{f\leq j\leq l}\bd(\sfs_j)$}{
		\ForEach{$(\leq 2)$-extension $\sfS'$ of $\sfS$ duplicating none, one or both of $\sfs_f$ and $\sfs_l$}{
			let $\ell'$ be the length of $\sfS'$\;
			set $f_x=\max\{j\in[\ell']\mid \delta_{\sfS'\rightarrow \sfS}(j)=f\}$ and $l_x=\min\{j\in[\ell']\mid \delta_{\sfS'\rightarrow \sfS}(j)=l\}$\;
			set $\sfS^x=\ins(\sfS',x,X,f_x,l_x)$\;
			(observe that by construction $(f_x,l_x)$ is valid with respect to $X$ in $\sfS'$)\;
			\lIf{$\mathsf{width}(\sfS^x)\leq \w$,}{add $\rep(\sfS^x)$ to $\cD_\w(\G^x)$}
			}
		}	
	}	
	\Return{ $\cD_\w(\G^x).$}
\caption{Insertion Routine} \label{alg:insertion}
}
\end{algorithm}

To prove the correctness of  Insertion Routine, we proceed in two steps. We first establish the \emph{completeness} of the algorithm. More precisely, 
Proposition~\ref{prop:insertion-complete} aims at proving that for every connected path-decomposition $\sfP^x$ of $\G^x,$ the $(\G^x,\sfP^x)$-encoding sequence
$T^x$ is dominated by some $B^x$-boundaried sequence $\sfS^x$ that can be computed from a $B$-boundaried sequence $\sfS$ belonging to $\cD_\w(\G).$
Then we argue about the \emph{soundness} of  Insertion Routine. Proposition~\ref{prop:insertion-sound} shows that if $\sfS^x$ is generated
from a $B$-boundaried sequence $\sfS\in \cD_\w(\G),$ then $\rep(\sfS^x)$ belongs to $\cD_\w(\G^x).$

\begin{proposition}[Insertion completeness] \label{prop:insertion-complete}

Let $\G=(G,B)$ be a boundaried graph and let $X\subseteq B$ be a subset of boundary vertices. Let $\sfP^x$ be a connected path-decomposition of width at most $\w$ of the boundaried graph $\G^x=(G^x,B^x)$ and let $\sfT^x$ be the $(\G^x,\sfP^x)$-encoding sequence. Then there exist a $B$-boundaried sequence $\sfS'$ such that $\sfS'$ is a $(\leq 2)$-extension of some $B$-boundaried sequence $\sfS\in\cD_\w(\G)$ and an insertion position $(f_x,l_x)$ valid with respect to $X$ in $\sfS'$ such that the $B^x$-boundaried sequence $\sfS^x=\ins(\sfS',x,X,f_x,l_x)$ satisfies $\rep(\sfS^x)\preceq \rep(\sfT^x).$

\end{proposition}
\begin{proof}
Suppose that $\sfP^x=\langle A^x_1,\dots, A^x_\ell\rangle$ and that $\sfT^x=\langle\sft^x_1,\dots, \sft^x_p\rangle.$ Let $[f^*_x,l^*_x]$ be the trace of $x$ in $\sfP^x.$ By the definition of a path-decomposition and of an encoding sequence, $X\subseteq \bigcup_{f_x\leq j\leq l_x}\bd(\sft^x_j).$ By Lemma \ref{lem:proj-boundaried}, $\sfP=\langle A_1,\dots, A_\ell\rangle,$ with $A_i=A_i^x\setminus\{x\}$ for every $1\leq i \leq\ell,$ is a connected path-decomposition of $\G.$  Let $\sfT=\langle \sft_1,\dots, \sft_\ell\rangle$ be the $(\G,\sfP)$-encoding sequence. Observe that by  the 
construction of $\G^x,$ if $y\in X,$ then $y\in A_j$ for some $f^*_x\leq j\leq l^*_x.$ As by 
assumption, $X\subseteq B,$ we have that $y\in \bd(\sft_j).$ Therefore, $(f^*_x,l^*_x)$ is 
a valid insertion position with respect to $X$ in $\sfT.$ One can easily check that $\ins(\sfT,x,X,f^*_x,l^*_x)=\sfT^x.$ Observe that, as the width of $\sfT^x$ is at most $\w,$ 
the width of $\sfT$ is at most $\w$ as well, because of Lemma~\autoref{lem:i-wd-ins}. Since $\cD_\w(\G)$ is a domination set of $\Rep_\w(\G),$ there exists a $B$-boundaried 
sequence $\sfS\in\cD_\w(\G)$ such that $\sfS\preceq \sfT.$ By Lemma~\ref{lem:i-completeness2}, there exists a $(\leq 2)$-extension $\sfS'$ of $\sfS$ and a valid 
insertion position $(f_x,l_x)$ with respect to $X$ in $\sfS'$ such that $\ins(\sfS',x,X,f_x,l_x)\preceq \ins(\sfT,x,X,f^*_x,l^*_x).$ By Lemma~\ref{lem:prop-useful-sequence}(3), we have $\rep(\sfS^x)\preceq \rep(\sfT^x).$
\end{proof}

{We let the reader observe that the completeness of  Insertion Routine relies on Lemma~\ref{lem:i-completeness2} and thereby on Lemma~\ref{lem:i-ext}. And the reason we compute a domination set of $\Rep_w(\G^x)$ rather than the set $\Rep_w(\G^x)$, is the issue discussed in Figure~\ref{fig:counter-example}.}

\begin{proposition}[Insertion soundness] \label{prop:insertion-sound}
Let $\G=(G,B)$ be a boundaried graph and let $X\subseteq B$ be a subset of boundary vertices. If $\sfS'=\langle\sfs'_1,\dots, \sfs'_{\ell'}\rangle$ is a $(\leq 2)$-extension of a $B$-boundaried sequence $\sfS=\langle \sfs_1,\dots, \sfs_{\ell}\rangle\in \cD_\w(\G)$ and if $(f_x,l_x)$ is a valid insertion position with respect to $X$ in $\sfS'$ such that $\sfS^x=\ins(\sfS',x,X,f_x,l_x)$ has width at most $\w,$ then $\rep(\sfS^x)\in\Rep_\w(\Gx).$
\end{proposition}
\begin{proof}
As $\sfS\in\cD_\w(\G)\subseteq \Rep_\w(\G),$ there exists a connected path-
decomposition $\sfP$ of $\G$ of width at most $\w$ such that the $(\G,\sfP)$-encoding 
sequence $\sfT=\langle\sft_1,\dots, \sft_p\rangle$ satisfies $\rep(\sfT)=\sfS.$ Let $\delta_{\sfS'\rightarrow\sfS}:[\ell']\rightarrow [\ell]$ be the extension surjection certifying 
that $\sfS'$ is a $(\leq 2)$-extension of $\sfS.$ Let us denote $f=\delta_{\sfS'\rightarrow \sfS}(f_x)$ and  $l=\delta_{\sfS'\rightarrow \sfS}(l_x).$ As $\sfS=\rep(\sfT),$ with every $j\in[\ell],$ we can associate a $\iota_j\in[p]$ such that $\sfS$ is the subsequence of $\sfT$ induced by $\bp(\sfT)=\{\iota_j\in[p]\mid j\in[\ell]\}.$ We build a $(\leq 2)$-extension 
$\sfT'=\langle \sft_1,\dots, \sft_{p'}\rangle$ of $\sfT,$ in the same way as $\sfS'$  is obtained 
from $\sfS,$ that is: we duplicate $\sft_{i_f}$ if and only if $\sfs_f$ is duplicated, and we 
duplicate $\sft_{i_j}$ if and only if $\sfs_l$ is duplicated. Observe that $\sfS'$ is the 
subsequence of $\sfT'$ induced by $\{i_{j}\in[p']\mid j\in[\ell']\}$ (see 
Figure~\ref{fig:insertion-soundness}). By construction of $\sfT',$ $(i_{f_x},i_{l_x})$ is a 
valid insertion position with respect to $X$ in $\sfT'.$ Thereby, we can define $\sfT^x=\ins(\sfT',x,X,i_{f_x},i_{l_x})$ and $\sfS^x=\ins(\sfS',x,X,f_x,l_x).$
Let $\sfP'$ be the connected path-decomposition obtained from $\sfP$ by duplicating the 
bags corresponding to $\sft_{\iota_f}$ and $\sft_{\iota_l}$ and adding $x$ to all bags 
between the bags associated with $\sft'_{i_{f_x}}$ and $\sft'_{i_{l_x}}.$
We remark that $\sfT^x$ is the $(\G^x,\sfP')$-encoding sequence and is thereby realizable.

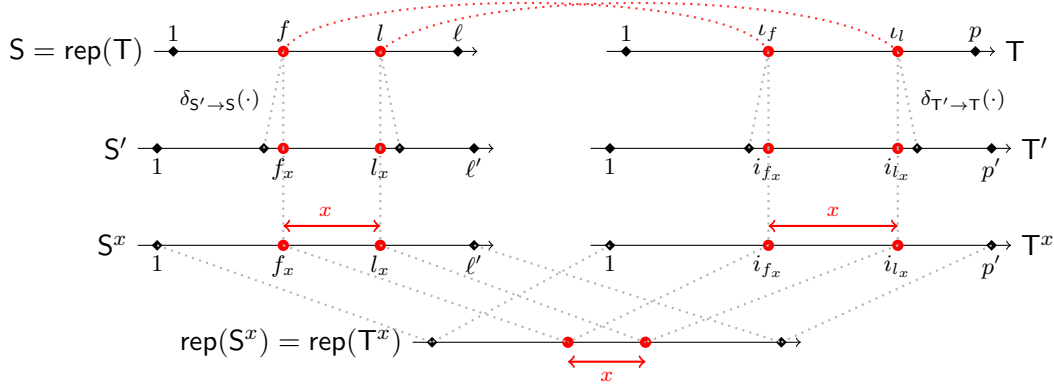
\begin{figure}[h]
\centering
\begin{tikzpicture}[scale=.86 ]
   \tikzstyle{vertex}=[fill,red,circle,minimum size=0.15cm,inner sep=0pt]
   \tikzstyle{vertex2}=[fill,black,diamond,minimum size=0.15cm,inner sep=0pt]

%------ \sfS -------
\draw[black,thin,->] (-7,3.5) -- (-2,3.5);
\node[anchor=east] at (-7,3.5) {$\sfS=\rep(\sfT)$};

\coordinate (1s) at (-6.7,3.5);
\coordinate (Fs) at (-5,3.5);
\coordinate (Ls) at (-3.5,3.5);
\coordinate (ls) at (-2.3,3.5);

\node[vertex2] (1s) at (1s) {};
\node[vertex] (f) at (Fs) {};
\node[vertex] (l) at (Ls) {};
\node[vertex2] (ls) at (ls) {};

\node[anchor=south] at (1s) {\footnotesize $1$};
\node[anchor=south] at (Fs) {\footnotesize $f$};
\node[anchor=south] at (Ls) {\footnotesize $l$};
\node[anchor=south] at (ls) {\footnotesize $\ell$};

%------ \sfS' -------
\draw[black,thin,->] (-7.25,2) -- (-1.75,2);
\node[anchor=east] at (-7.25,2) {$\sfS'$};

\coordinate (1S') at (-6.95,2);
\coordinate (FS') at (-5.3,2);
\coordinate (FS'X) at (-5,2);
\coordinate (LS'X) at (-3.5,2);
\coordinate (LS') at (-3.2,2);
\coordinate (lS') at (-2.05,2);

\node[vertex2] (1s') at (1S'){};
\node[vertex2] (f') at (FS'){};
\node[vertex] (fx') at (FS'X){};
\node[vertex] (lx') at (LS'X){};
\node[vertex2] (l') at (LS'){};
\node[vertex2] (ls') at (lS'){};

\node[anchor=north] at (1S') {\footnotesize $1$};
\node[anchor=north] at (FS'X) {\footnotesize $f_x$};
\node[anchor=north] at (LS'X) {\footnotesize $l_x$};
\node[anchor=north] at (lS') {\footnotesize $\ell'$};

%------ \sfSx-------
\draw[black,thin,->] (-7.25,0.5) -- (-1.75,0.5);
\node[anchor=east] at (-7.25,0.5) {$\sfS^x$};

\coordinate (1Sx) at (-6.95,0.5);
\coordinate (FSx) at (-5.3,0.5);
\coordinate (FSxX) at (-5,0.5);
\coordinate (LSxX) at (-3.5,0.5);
\coordinate (LSx) at (-3.2,0.5);
\coordinate (lSx) at (-2.05,0.5);

\node[vertex2] (1sx) at (1Sx){};
\node[vertex] (fxx) at (FSxX){};
\node[vertex] (lxx) at (LSxX){};
\node[vertex2] (lsx) at (lSx){};

\node[anchor=north] at (1Sx) {\footnotesize $1$};
\node[anchor=north] at (FSxX) {\footnotesize $f_x$};
\node[anchor=north] at (LSxX) {\footnotesize $l_x$};
\node[anchor=north] at (lSx) {\footnotesize $\ell'$};

\draw[red,thick,<->] (-5,0.8) -- (-3.5,0.8);
\node[anchor=south,red] at (-4.35,0.8) {\scriptsize $x$};

%----- \delta S'-S -----
\draw[gray!60,thick,dotted] (FS') -- (Fs);
\draw[gray!60,thick,dotted] (FS'X) -- (Fs);
\draw[gray!60,thick,dotted] (LS'X) -- (Ls);
\draw[gray!60,thick,dotted] (LS') -- (Ls);

\node[anchor=east] at (-5.2,2.75) {\scriptsize $\delta_{\sfS'\rightarrow\sfS}(\cdot)$};

\draw[gray!60,thick,dotted] (FS'X) -- (FSxX);
\draw[gray!60,thick,dotted] (LS'X) -- (LSxX);

%--------------------------------------------------------
%--------------------------------------------------------
%------ \sfT -------
\draw[black,thin,->] (0,3.5) -- (6,3.5);
\node[anchor=west] at (6,3.5) {$\sfT$};

\coordinate (1t) at (0.3,3.5);
\coordinate (Ft) at (2.5,3.5);
\coordinate (Lt) at (4.5,3.5);
\coordinate (pt) at (5.7,3.5);

\node[vertex2] (1t) at (1t) {};
\node[vertex] (ft) at (Ft) {};
\node[vertex] (lt) at (Lt) {};
\node[vertex2] (pt) at (pt) {};

\node[anchor=south] at (1t) {\footnotesize $1$};
\node[anchor=south] at (Ft) {\footnotesize $\iota_f$};
\node[anchor=south] at (Lt) {\footnotesize $\iota_l$};
\node[anchor=south] at (pt) {\footnotesize $p$};

%------ \sfT' -------
\draw[black,thin,->] (-0.25,2) -- (6.25,2);
\node[anchor=west] at (6.25,2) {$\sfT'$};

\coordinate (1T') at (0.05,2);
\coordinate (FT') at (2.2,2);
\coordinate (FT'X) at (2.5,2);
\coordinate (LT'X) at (4.5,2);
\coordinate (LT') at (4.8,2);
\coordinate (pT') at (5.95,2);

\node[vertex2] (1s') at (1T'){};
\node[vertex2] (f') at (FT'){};
\node[vertex] (fx') at (FT'X){};
\node[vertex] (lx') at (LT'X){};
\node[vertex2] (l') at (LT'){};
\node[vertex2] (ps') at (pT'){};

\node[anchor=north] at (1T') {\footnotesize $1$};
\node[anchor=north] at (FT'X) {\footnotesize $i_{f_x}$};
\node[anchor=north] at (LT'X) {\footnotesize $i_{l_x}$};
\node[anchor=north] at (pT') {\footnotesize $p'$};

%------ \sfTx-------
\draw[black,thin,->] (-0.25,0.5) -- (6.25,0.5);
\node[anchor=west] at (6.25,0.5) {$\sfT^x$};

\coordinate (1Tx) at (0.05,0.5);
\coordinate (FTx) at (2.2,0.5);
\coordinate (FTxX) at (2.5,0.5);
\coordinate (LTxX) at (4.5,0.5);
\coordinate (LTx) at (4.8,0.5);
\coordinate (pTx) at (5.95,0.5);

\node[vertex2] (1tx) at (1Tx){};
\node[vertex] (ftxx) at (FTxX){};
\node[vertex] (ltxx) at (LTxX){};
\node[vertex2] (ptx) at (pTx){};

\node[anchor=north] at (1Tx) {\footnotesize $1$};
\node[anchor=north] at (FTxX) {\footnotesize $i_{f_x}$};
\node[anchor=north] at (LTxX) {\footnotesize $i_{l_x}$};
\node[anchor=north] at (pTx) {\footnotesize $p'$};

\draw[red,thick,<->] (2.5,0.8) -- (4.5,0.8);
\node[anchor=south,red] at (3.5,0.8) {\scriptsize $x$};

%----- \delta T'-T -----
\draw[gray!60,thick,dotted] (FT') -- (Ft);
\draw[gray!60,thick,dotted] (FT'X) -- (Ft);
\draw[gray!60,thick,dotted] (LT'X) -- (Lt);
\draw[gray!60,thick,dotted] (LT') -- (Lt);

\node[anchor=west] at (4.7,2.75) {\scriptsize $\delta_{\sfT'\rightarrow\sfT}(\cdot)$};

\draw[gray!60,thick,dotted] (FT'X) -- (FTxX);
\draw[gray!60,thick,dotted] (LT'X) -- (LTxX);

%------ Rep(\sfTx)=Rep(\sfSx)-------
\draw[black,thin,->] (-3,-1) -- (3,-1);
\node[anchor=east] at (-3,-1) {$\rep(\sfS^x)=\rep(\sfT^x)$};

\coordinate (R1) at (-2.7,-1);
\coordinate (FR) at (-0.6,-1);
\coordinate (LR) at (0.6,-1);
\coordinate (Rn) at (2.7,-1);

\node[vertex2] (r1) at (R1){};
\node[vertex] (rf) at (FR){};
\node[vertex] (rl) at (LR){};
\node[vertex2] (rn) at (Rn){};

\draw[red,thick,<->] (-0.6,-1.3) -- (0.6,-1.3);
\node[anchor=north,red] at (0,-1.3) {\scriptsize $x$};

%-------------
\draw[gray!60,thick,dotted] (1Tx) -- (R1);
\draw[gray!60,thick,dotted] (FTxX) -- (FR);
\draw[gray!60,thick,dotted] (LTxX) -- (LR);
\draw[gray!60,thick,dotted] (pTx) -- (Rn);

\draw[gray!60,thick,dotted] (1Sx) -- (R1);
\draw[gray!60,thick,dotted] (FSxX) -- (FR);
\draw[gray!60,thick,dotted] (LSxX) -- (LR);
\draw[gray!60,thick,dotted] (lSx) -- (Rn);

\draw[red!80,thick,dotted] (Fs) .. controls (-4.5,4.5) and (2,4.5) .. (Ft) ;
\draw[red!80,thick,dotted] (Ls) .. controls (-2.5,4.5) and (3.5,4.5) .. (Lt) ;

\end{tikzpicture}
\caption{Soundness of the insertion routine: if $\sfS'=\langle\sfs'_1,\dots, \sfs'_{\ell'}\rangle$ is a $(\leq 2)$-extension of a $B$-boundaried sequence $\sfS=\rep(\sfT)\in \cD_w(\G)$ and $(f_x,l_x)$ is a valid insertion position with respect to $X$ in $\sfS'$, then $\rep(\sfS^x)\in \cD_w(\G^x)$.
\label{fig:insertion-soundness}}
\end{figure}

We claim now that $\rep(\sfS^x) = \rep(\sfT^x).$ Because $\sfS=\rep(\sfT),$ one can prove, in the same way as the second statement of Lemma
\ref{lem:prop-useful}(6), that there are $\sfS_1$ and $\sfS_2,$ extensions of $\sfS,$ such that $\sfS_1\leq \sfT \leq \sfS_2,$
  $\delta_{\sfS_1\to \sfS}(i_j) = \delta_{\sfS_2\to \sfS}(i_j) = j\in [\ell],$ and
  $i_{f_x} = \min\{h\in [p]\mid f=\delta_{\sfS_1\rightarrow\sfS}(h)=\delta_{\sfS_2\rightarrow\sfS}(h)\}$ and
  $i_{l_x}=\max\{h\in [p]\mid l=\delta_{\sfS_1\rightarrow\sfS}(h) = \delta_{\sfS_2\rightarrow\sfS}(h)\}.$ By making the same duplications
  in $\sfS'$ as in $\sfS$ to obtain $\sfS_1$ and $\sfS_2,$ one can construct extensions $\sfS'_1$ and $\sfS'_2$ of $\sfS'$ such that
  $\sfS'_1 \leq \sfT' \leq \sfS'_2,$ $\delta_{\sfS'_1\to \sfS'}(i_j) = \delta_{\sfS'_2\to \sfS'}(i_j) = j\in [\ell'],$ and
  $i_{f_x} = \min\{h\in [p']\mid f_x=\delta_{\sfS'_1\rightarrow\sfS'}(h)=\delta_{\sfS'_2\rightarrow\sfS'}(h)\}$ and
  $i_{l_x}=\max\{h\in [p']\mid l_x=\delta_{\sfS'_1\rightarrow\sfS'}(h) = \delta_{\sfS'_2\rightarrow\sfS'}(h)\}.$ Therefore,
  $(i_{f_x},i_{l_x})$ is a valid insertion position with respect to $X$ in both $\sfS'_1$ and $\sfS'_2.$ By Lemma \ref{lem:i-ext2}, we have
  $\ins(\sfS'_1,x,X,i_{f_x},i_{l_x}) \leq \sfT^x \leq \ins(\sfS'_2,x,X,i_{f_x},i_{l_x}).$ Because $\sfS'_1$ and $\sfS'_2$ are both
  extensions of $\sfS’$, $i_{f_x} = \min\{h\in [p']\mid f_x=\delta_{\sfS'_1\rightarrow\sfS'}(h)=\delta_{\sfS'_2\rightarrow\sfS'}(h)\}$, and $i_{l_x}=\max\{h\in [p']\mid l_x=\delta_{\sfS'_1\rightarrow\sfS'}(h)\} = \delta_{\sfS'_2\rightarrow\sfS'}(h)\},$ we can conclude by
  Lemma \ref{lem:i-val-posb} that $\ins(\sfS'_1,x,X,i_{f_x},i_{l_x})$ and $\ins(\sfS'_2,x,X,i_{f_x},i_{l_x})$ are both extensions of
  $\sfS^x.$ We can therefore conclude that $\sfS^x\equiv \sfT^x,$ \ie, $\rep(\sfS^x)=\rep(\sfT^x).$ Finally, as $\sfT^x$ is realisable, we
  can conclude that $\rep(\sfS^x) \in\Rep_\w(\Gx).$
\end{proof}

\begin{theorem} \label{th:insertion}
Algorithm~\ref{alg:insertion} computes $\cD_\w(\G^x)$  in $2^{O(k (w + \log k))}$-time, where $k=|B|.$
%\sed{Explain parametric dependences}
\end{theorem}
\begin{proof}
The correctness of Algorithm~\ref{alg:insertion} is proved by Proposition~\ref{prop:insertion-complete} and
Proposition~\ref{prop:insertion-sound}. These two propositions imply that by applying Insertion Routine on a domination set of ${\bf G}$
that is a subset of the representatives of $\G$, we indeed compute a domination set of $\G^x$ that is a subset of the set of representatives of $\G^x$. Let us
analyse its time complexity. By Lemma~\ref{lem:nb-rep}, the size of $\Rep_\w(\G)$ (and so the size of $\cD_\w(\G)$) depends on $k$ and $\w.$ By
Lemma~\ref{lem:size-model}, the length of a representative $B$-boundaried sequence of $\Rep_\w(\G)$ depends on $k.$ As performing the insertion in a
$B$-boundaried sequence can be performed in polynomial time in the size of the sequence, the time complexity of Algorithm~\ref{alg:insertion} is dominated by the size of
$\cD_\w(\G)$ that is $2^{O(k (w + \log k))}$, because of Lemma \autoref{lem:nb-rep}.
\end{proof}

%-----------------------
\subsection{The dynamic programming algorithm}

We are now in position to prove \autoref{main_erd}. 
We first explain  an algorithm that decides whether ${\sf cpw}(G)\leq w$.  
Suppose that we are given a path-decompositon
$\mathsf{Q}=\langle B_1,\dots, B_q\rangle$ of $G$ of width at most $k.$
Our algorithm  performs dynamic programming   over $\mathsf{Q}.$ 
For each $i\in[q]$, we consider the boundaried graph ${\bf G}_{i}=(G[V_{i}],B_{i})$, where $V_i=\bigcup_{1\leq h\leq i}B_h.$ The task is to compute for every $i\in[q],$ a domination set $\cD_{\w+1}(\G_i).$ Let us describe $\cD_{\w+1}(\G_1).$ As $\mathsf{Q}$ is a nice path-decomposition, $B_1=\{x\}$ for some $x\in V.$ The representative  set $\Rep_{\w+1}(\G_1)$ consists for the following  four possible connected $B_1$-boundaried sequences:
\begin{itemize}
\item $\sfS_1=\langle (\{x\},\{\{x\}\},0)\rangle,$
\item $\sfS_2=\langle (\emptyset,\{\emptyset\},0), (\{x\},\{\{x\}\},0)\rangle,$
\item $\sfS_3=\langle (\emptyset,\{\emptyset\},0), (\{x\},\{\{x\}\},0), (\emptyset,\{\{x\}\},0)\rangle,$ and 
\item $\sfS_4=\langle (\{x\},\{\{x\}\},0), (\emptyset,\{\{x\}\},0) \rangle.$
\end{itemize} 
We use $\Rep_{\w+1}(\G_1)$, as $\cD_{\w+1}(\G_1)$ as none of the above sequence is dominating the other.
Now Algorithm~\ref{alg:insertion} and Algorithm~\ref{alg:forget} describe how to compute for every $1<i\leq q,$ $\cD_{\w+1}(\G_i)$ depending on whether $B_i$ is an insertion or a forgetting bag. We obtain that $\cpw(G)\leq{\w}$ if and only if $\cD_{\w+1}(\G_q)\neq\emptyset$, because of Proposition~\autoref{prop:correctness}. The correctness of the DP algorithm described above follows from Theorem~\ref{th:forget}, Theorem~\ref{th:insertion}. The time complexity depends on the running time of  Insertion Routine (Algorithm~\ref{alg:insertion}) and  Forget Routine (Algorithm~\ref{alg:forget}) described respectively in  Theorem~\ref{th:forget} and Theorem~\ref{th:insertion}. We just proved  the decision version of \autoref{main_erd}. 
In \cite[Section 6]{BodlaenderK96effi} Bodlaender and Kloks explained how to turn their decision algorithm for pathwidth and treewidth to one that is able to construct, in case of a positive answer, 
the corresponding decomposition. Following the same arguments, it is straightforward to 
transform the above decision algorithm for connected pathwidth to one that also constructs the 
connected path-decomposition, if it exists. This completes the proof of \autoref{main_erd}.

\begin{theorem}
One may  construct an algorithm that, given an $n$-connected graph $G$ and a non-negative integer $k$, either 
outputs a connected path-decomposition of $G$ of width at most $k$ or correctly 
reports that such a decomposition does not exist in $2^{O(k^2)}\cdot n$ time.
\end{theorem}
\begin{proof}
According to the result of Fürer\cite{Furer16fast} there is an algorithm that, given a graph $G$
and an integer $k$, outputs, if exists, a path-decomposition of width at most $k$ in  $2^{O(k^2)}\cdot n$ time. We run this algorithm and if
the answer is negative, we report that ${\sf cpw}(G)>k$ and we are done (here we use Observation~\ref{obs:cpw}). Otherwise we use the
provided path-decomposition in order to solve the problem in $2^{O(w(k+\log w))}\cdot n$ time using the algorithm of \autoref{main_erd}
where $w\leq k$ is the width of the constructed path-decomposition in the first step.
\end{proof}

%%------------------------------------------------------------------------------------------------------------------------


\begin{thebibliography}{10}

\bibitem{AdlerPT21conn}
Isolde Adler, Christophe Paul, and Dimitrios~M. Thilikos.
\newblock Connected search for a lazy robber.
\newblock {\em Journal of Graph Theory}, 97:510--552, 2021.
\newblock \href {http://dx.doi.org/10.1002/jgt.22669}
  {\path{doi:10.1002/jgt.22669}}.

\bibitem{Fomin17comp}
Spyros Angelopoulos, Pierre Fraigniaud, Fedor~V. Fomin, Nicolas Nisse, and
  Dimitrios~M. Thilikos.
\newblock Report on {GRASTA} 2017, 6th workshop on graph searching, theory and
  applications.
\newblock Technical Report {HAL} lirmm-01645614, CNRS, Universit\'e
  Montpellier, LIRMM, 2017.
\newblock URL: \url{https://hal-lirmm.ccsd.cnrs.fr/lirmm-01645614/document}.

\bibitem{ArnborgCP87comp}
Stefan Arnborg, Derek~G. Corneil, and Andrzej Proskurowski.
\newblock Complexity of finding embeddings in a {$k$}-tree.
\newblock {\em SIAM Journal on Algebraic and Discrete Methods}, 8(2):277--284,
  1987.
\newblock \href {http://dx.doi.org/10.1137/0608024}
  {\path{doi:10.1137/0608024}}.

\bibitem{BarriereFFFNST12conn}
Lali Barri{\`e}re, Paola Flocchini, Fedor~V. Fomin, Pierre Fraigniaud, Nicolas
  Nisse, Nicola Santoro, and Dimitrios~M. Thilikos.
\newblock Connected graph searching.
\newblock {\em Information and Computation}, 219:1--16, 2012.
\newblock \href {http://dx.doi.org/10.1016/j.ic.2012.08.004}
  {\path{doi:10.1016/j.ic.2012.08.004}}.

\bibitem{BarriereFFS02capt}
Lali Barri{\`e}re, Paola Flocchini, Pierre Fraigniaud, and Nicola Santoro.
\newblock Capture of an intruder by mobile agents.
\newblock In {\em Annual ACM Symposium on Parallel Algorithms and Architectures
  ({SPAA})}, pages 200--209, 2002.
\newblock \href {http://dx.doi.org/10.1145/564870.564906}
  {\path{doi:10.1145/564870.564906}}.

\bibitem{BarriereFST03sear}
Lali Barri{\`{e}}re, Pierre Fraigniaud, Nicola Santoro, and Dimitrios~M.
  Thilikos.
\newblock Searching is not jumping.
\newblock In {\em International Workshop Graph-Theoretic Concepts in Computer
  Science, ({WG})}, volume 2880 of {\em Lecture Notes in Computer Science},
  pages 34--45, 2003.
\newblock \href {http://dx.doi.org/10.1007/978-3-540-39890-5\_4}
  {\path{doi:10.1007/978-3-540-39890-5\_4}}.

\bibitem{BienstockS91mono}
D.~Bienstock and Paul~D. Seymour.
\newblock Monotonicity in graph searching.
\newblock {\em Journal of Algorithms}, 12(2):239--245, 1991.
\newblock \href {http://dx.doi.org/10.1016/0196-6774(91)90003-H}
  {\path{doi:10.1016/0196-6774(91)90003-H}}.

\bibitem{BienstockRST91quic}
Dan Bienstock, Neil Robertson, Paul~D. Seymour, and Robin Thomas.
\newblock Quickly excluding a forest.
\newblock {\em Journal of Combinatorial Theory, Series B}, 52(2):274--283,
  1991.
\newblock \href {http://dx.doi.org/10.1016/0095-8956(91)90068-U}
  {\path{doi:10.1016/0095-8956(91)90068-U}}.

\bibitem{Bienstock89grap}
Daniel Bienstock.
\newblock Graph searching, path-width, tree-width and related problems (a
  survey).
\newblock In {\em Reliability of computer and communication networks}, volume~5
  of {\em {DIMACS} Series in Discrete Mathematics and Theoretical Computer
  Science}, pages 33--50, 1991.

\bibitem{BodlaenderFT09deri}
Hans~L. Bodlaender, Michael~R. Fellows, and Dimitrios~M. Thilikos.
\newblock Derivation of algorithms for cutwidth and related graph layout
  parameters.
\newblock {\em Journal of Computer and System Sciences}, 75(4):231--244, 2009.
\newblock \href {http://dx.doi.org/10.1016/j.jcss.2008.10.003}
  {\path{doi:10.1016/j.jcss.2008.10.003}}.

\bibitem{BodlaenderJT20}
Hans~L. Bodlaender, Lars Jaffke, and Jan~Arne Telle.
\newblock Typical sequences revisited - computing width parameters of graphs.
\newblock In {\em 37th International Symposium on Theoretical Aspects of
  Computer Science, ({STACS})}, volume 154 of {\em Leibniz International
  Proceedings in Informatics}, pages 57:1--57:16, 2020.
\newblock \href {http://dx.doi.org/10.4230/LIPIcs.STACS.2020.57}
  {\path{doi:10.4230/LIPIcs.STACS.2020.57}}.

\bibitem{BodlaenderK96effi}
Hans~L. Bodlaender and Ton Kloks.
\newblock Efficient and constructive algorithms for the pathwidth and treewidth
  of graphs.
\newblock {\em Journal of Algorithms}, 21(2):358--402, 1996.
\newblock \href {http://dx.doi.org/10.1006/jagm.1996.0049}
  {\path{doi:10.1006/jagm.1996.0049}}.

\bibitem{BodlaenderT97cons}
Hans~L. Bodlaender and Dimitrios~M. Thilikos.
\newblock Constructive linear time algorithms for branchwidth.
\newblock In {\em International Colloquium Automata, Languages and Programming,
  ({ICALP})}, volume 1256 of {\em Lecture Notes in Computer Science}, pages
  627--637, 1997.
\newblock \href {http://dx.doi.org/10.1007/3-540-63165-8_217}
  {\path{doi:10.1007/3-540-63165-8_217}}.

\bibitem{BodlaenderT98comp}
Hans~L. Bodlaender and Dimitrios~M. Thilikos.
\newblock Computing small search numbers in linear time.
\newblock In {\em International Workshop on Parameterized and Exact
  Computation, ({IWPEC})}, volume 3162 of {\em Lecture Notes in Computer
  Science}, pages 37--48, 2004.
\newblock \href {http://dx.doi.org/10.1007/978-3-540-28639-4\_4}
  {\path{doi:10.1007/978-3-540-28639-4\_4}}.

\bibitem{BojanczykP17opti}
Miko{\l}aj Boja{\'{n}}czyk and Micha{\l} Pilipczuk.
\newblock Optimizing tree decompositions in {MSO}.
\newblock In {\em International Symposium on Theoretical Aspects of Computer
  Science, ({STACS})}, volume~66 of {\em Leibniz International Proceedings in
  Informatics}, pages 15:1--15:13, 2017.
\newblock \href {http://dx.doi.org/10.4230/LIPIcs.STACS.2017.15}
  {\path{doi:10.4230/LIPIcs.STACS.2017.15}}.

\bibitem{Breisch67anin}
R.~Breisch.
\newblock An intuitive approach to speleotopology.
\newblock {\em Southwestern Cavers (A publication of ~the Southwestern Region
  of the National Speleological Society)}, VI(5):72--78, 1967.

\bibitem{ChartrandZHHMB03grap}
Gary Chartrand, Ping Zhang, Teresa~W. Haynes, Michael~A. Henning, Fred~R.
  McMorris, and Robert~C. Brigham.
\newblock {\em Graphical Measurement}, chapter~9, pages 872--951.
\newblock Discrete Mathematics and Its Applications. Chapman {\&} Hall / Taylor
  {\&} Francis, 2003.
\newblock \href {http://dx.doi.org/10.1201/9780203490204}
  {\path{doi:10.1201/9780203490204}}.

\bibitem{Courcelle90them}
Bruno Courcelle.
\newblock The monadic second-order logic of graphs. {I}. recognizable sets of
  finite graphs.
\newblock {\em Information and Computation}, 85(1):12--75, 1990.
\newblock \href {http://dx.doi.org/10.1016/0890-5401(90)90043-H}
  {\path{doi:10.1016/0890-5401(90)90043-H}}.

\bibitem{CourcelleL96equi}
Bruno Courcelle and Jens Lagergren.
\newblock Equivalent definitions of recognizability for sets of graphs of
  bounded tree-width.
\newblock {\em Mathematical Structures in Computer Science}, 6(2):141--165,
  1996.
\newblock \href {http://dx.doi.org/10.1017/S096012950000092X}
  {\path{doi:10.1017/S096012950000092X}}.

\bibitem{Dereniowski12from}
Dariusz Dereniowski.
\newblock From pathwidth to connected pathwidth.
\newblock {\em SIAM Journal on Discrete Mathematics}, 26(4):1709--1732, 2012.
\newblock \href {http://dx.doi.org/10.1137/110826424}
  {\path{doi:10.1137/110826424}}.

\bibitem{DereniowskiOR19find}
Dariusz Dereniowski, Dorota Osula, and Pawe{\l} Rz{\k{a}}{\.{z}}ewski.
\newblock Finding small-width connected path decompositions in polynomial time.
\newblock {\em Theoretical Computer Science}, 794:85--100, 2019.
\newblock \href {http://dx.doi.org/10.1016/j.tcs.2019.03.039}
  {\path{doi:10.1016/j.tcs.2019.03.039}}.

\bibitem{FominT01onth}
Fedor~V. Fomin and Dimitrios~M. Thilikos.
\newblock On the monotonicity of games generated by symmetric submodular
  functions.
\newblock {\em Discrete Applied Mathematics}, 131(2):323--335, 2003.
\newblock \href {http://dx.doi.org/10.1016/S0166-218X(02)00459-6}
  {\path{doi:10.1016/S0166-218X(02)00459-6}}.

\bibitem{FominT08anan}
Fedor~V. Fomin and Dimitrios~M. Thilikos.
\newblock An annotated bibliography on guaranteed graph searching.
\newblock {\em Theoretical Computer Science}, 399(3):236--245, 2008.
\newblock \href {http://dx.doi.org/10.1016/j.tcs.2008.02.040}
  {\path{doi:10.1016/j.tcs.2008.02.040}}.

\bibitem{Furer16fast}
Martin F{\"{u}}rer.
\newblock Faster computation of path-width.
\newblock In {\em International Workshop on Combinatorial Algorithms,
  ({IWOCA})}, volume 9843 of {\em Lecture Notes in Computer Science}, pages
  385--396, 2016.
\newblock \href {http://dx.doi.org/10.1007/978-3-319-44543-4\_30}
  {\path{doi:10.1007/978-3-319-44543-4\_30}}.

\bibitem{Golovach89equi}
Petr~A. Golovach.
\newblock Equivalence of two formalizations of a search problem on a graph
  ({R}ussian).
\newblock {\em Vestnik Leningrad. Univ. Mat. Mekh. Astronom.}, vyp. 1:10--14,
  122, 1989.
\newblock translation in Vestnik Leningrad Univ. Math. 22 (1989), no. 1,
  13--19.

\bibitem{Jeong0O16cons}
Jisu Jeong, Eun~Jung Kim, and Sang{-}il Oum.
\newblock Constructive algorithm for path-width of matroids.
\newblock In {\em Annual {ACM-SIAM} Symposium on Discrete Algorithms,
  ({SODA})}, pages 1695--1704, 2016.
\newblock \href {http://dx.doi.org/10.1137/1.9781611974331.ch116}
  {\path{doi:10.1137/1.9781611974331.ch116}}.

\bibitem{JeongKO17the}
Jisu Jeong, Eun~Jung Kim, and Sang{-}il Oum.
\newblock The ``art of trellis decoding'' is fixed-parameter tractable.
\newblock {\em IEEE Transactions on Information Theory}, 63(11):7178--7205,
  2017.
\newblock \href {http://dx.doi.org/10.1109/TIT.2017.2740283}
  {\path{doi:10.1109/TIT.2017.2740283}}.

\bibitem{Jeong0O18find}
Jisu Jeong, Eun~Jung Kim, and Sang{-}il Oum.
\newblock Finding branch-decomposition of matroids, hypergraphs and more.
\newblock In {\em International Colloquium Automata, Languages and Programming,
  ({ICALP})}, volume 107 of {\em Leibniz International Proceedings in
  Informatics}, pages 80:1--80:14, 2018.
\newblock \href {http://dx.doi.org/10.4230/LIPIcs.ICALP.2018.80}
  {\path{doi:10.4230/LIPIcs.ICALP.2018.80}}.

\bibitem{KantePT19ali}
Mamadou~Moustapha Kant{\'{e}}, Christophe Paul, and Dimitrios~M. Thilikos.
\newblock A linear fixed parameter tractable algorithm for connected pathwidth.
\newblock In {\em Annual European Symposium on Algorithms, {ESA}}, volume 173
  of {\em Leibniz International Proceedings in Informatics}, pages 64:1--64:16,
  2020.
\newblock \href {http://dx.doi.org/10.4230/LIPIcs.ESA.2020.64}
  {\path{doi:10.4230/LIPIcs.ESA.2020.64}}.

\bibitem{Kinnersley92thev}
Nancy~G. Kinnersley.
\newblock The vertex separation number of a graph equals its path-width.
\newblock {\em Information Processing Letters}, 42(6):345--350, 1992.
\newblock \href {http://dx.doi.org/10.1016/0020-0190(92)90234-M}
  {\path{doi:10.1016/0020-0190(92)90234-M}}.

\bibitem{KirousisP85inte}
Lefteris~M. Kirousis and Christos~H. Papadimitriou.
\newblock Interval graphs and searching.
\newblock {\em Discrete Mathematics}, 55(2):181--184, 1985.
\newblock \href {http://dx.doi.org/10.1016/0012-365X(85)90046-9}
  {\path{doi:10.1016/0012-365X(85)90046-9}}.

\bibitem{KirousisP86sear}
Lefteris~M. Kirousis and Christos~H. Papadimitriou.
\newblock Searching and pebbling.
\newblock {\em Theoretical Computer Science}, 47(2):205--218, 1986.
\newblock \href {http://dx.doi.org/10.1016/0304-3975(86)90146-5}
  {\path{doi:10.1016/0304-3975(86)90146-5}}.

\bibitem{Lagergren98upp}
J.~Lagergren.
\newblock Upper bounds on the size of obstructions and intertwines.
\newblock {\em Journal of Combinatorial Theory, Series B}, 73:7--40, 1998.
\newblock \href {http://dx.doi.org/10.1006/jctb.1997.1788}
  {\path{doi:10.1006/jctb.1997.1788}}.

\bibitem{LagergrenA91find}
Jens Lagergren and Stefan Arnborg.
\newblock Finding minimal forbidden minors using a finite congruence.
\newblock In {\em International Colloquium on Automata, Languages and
  Programming, ({ICALP})}, volume 510 of {\em Lecture Notes in Computer
  Science}, pages 532--543, 1991.
\newblock \href {http://dx.doi.org/10.1007/3-540-54233-7\_161}
  {\path{doi:10.1007/3-540-54233-7\_161}}.

\bibitem{MescoffPT21apol}
Guillaume Mescoff, Christophe Paul, and Dimitrios~M. Thilikos.
\newblock A polynomial time algorithm to compute the connected treewidth of a
  series--parallel graph.
\newblock {\em Discrete Applied Mathematics}, 2021.
\newblock \href {http://dx.doi.org/10.1016/j.dam.2021.02.039}
  {\path{doi:10.1016/j.dam.2021.02.039}}.

\bibitem{Moehring90grap}
Rolf~H. M{\"o}hring.
\newblock Graph problems related to gate matrix layout and {PLA} folding.
\newblock In {\em Computational graph theory}, volume~7 of {\em Computing
  Supplementum}, pages 17--51. Springer, 1990.
\newblock \href {http://dx.doi.org/10.1007/978-3-7091-9076-0_2}
  {\path{doi:10.1007/978-3-7091-9076-0_2}}.

\bibitem{Pardo13purs}
Ronan Pardo~Soares.
\newblock {\em {Pursuit-Evasion, Decompositions and Convexity on Graphs}}.
\newblock PhD thesis, {Universit{\'e} Nice Sophia Antipolis}, 2013.
\newblock URL: \url{https://tel.archives-ouvertes.fr/tel-00908227}.

\bibitem{Parsons78purs}
Torrence~D. Parsons.
\newblock Pursuit-evasion in a graph.
\newblock In {\em International Conference on the Theory and Applications of
  Graphs}, volume 642 of {\em Lecture Notes in Mathematics}, pages 426--441,
  1978.

\bibitem{Parsons78thes}
Torrence~D. Parsons.
\newblock The search number of a connected graph.
\newblock In {\em Southeastern Conference on Combinatorics, Graph Theory, and
  Computing}, Congressus Numerantium, XXI, pages 549--554. Utilitas
  Mathematica, 1978.

\bibitem{Petrov82apro}
Nicolai~N. Petrov.
\newblock A problem of pursuit in the absence of information on the pursued.
\newblock {\em Differentsial'nye Uravneniya}, 18(8):1345--1352, 1468, 1982.

\bibitem{RobertsonS04GMXX}
Neil Robertson and P.~D. Seymour.
\newblock {Graph Minors. XX. Wagner's conjecture}.
\newblock {\em Journal of Combinatorial Theory, Series B}, 92(2):325--357,
  2004.
\newblock \href {http://dx.doi.org/10.1016/j.jctb.2004.08.001}
  {\path{doi:10.1016/j.jctb.2004.08.001}}.

\bibitem{RobertsonS83GM_I}
Neil Robertson and Paul~D. Seymour.
\newblock {Graph Minors. {I}. Excluding a forest}.
\newblock {\em Journal of Combinatorial Theory, Series B}, 35(1):39--61, 1983.
\newblock \href {http://dx.doi.org/10.1016/0095-8956(83)90079-5}
  {\path{doi:10.1016/0095-8956(83)90079-5}}.

\bibitem{Telle05tree}
Jan~Arne Telle.
\newblock Tree-decomposition of small pathwidth.
\newblock {\em Discrete Applied Mathematics}, 145(2):210--218, 2005.
\newblock \href {http://dx.doi.org/10.1016/j.dam.2004.01.012}
  {\path{doi:10.1016/j.dam.2004.01.012}}.

\bibitem{ThilikosSB00cons}
Dimitrios~M. Thilikos, Maria~J. Serna, and Hans~L. Bodlaender.
\newblock Constructive linear time algorithms for small cutwidth and
  carving-width.
\newblock In {\em International Symposium on Algorithms and computation
  ({ISAAC})}, volume 1969 of {\em Lecture Notes in Computer Science}, pages
  192--203, 2000.
\newblock \href {http://dx.doi.org/10.1007/3-540-40996-3_17}
  {\path{doi:10.1007/3-540-40996-3_17}}.

\bibitem{ThilikosSB05cutI}
Dimitrios~M. Thilikos, Maria~J. Serna, and Hans~L. Bodlaender.
\newblock Cutwidth. {I}. {A} linear time fixed parameter algorithm.
\newblock {\em Journal of Algorithms}, 56(1):1--24, 2005.
\newblock \href {http://dx.doi.org/10.1016/j.jalgor.2004.12.001}
  {\path{doi:10.1016/j.jalgor.2004.12.001}}.

\bibitem{ThilikosSB05cutII}
Dimitrios~M. Thilikos, Maria~J. Serna, and Hans~L. Bodlaender.
\newblock Cutwidth {II:} algorithms for partial w-trees of bounded degree.
\newblock {\em Journal of Algorithms}, 56(1):25--49, 2005.
\newblock \href {http://dx.doi.org/10.1016/j.jalgor.2004.12.003}
  {\path{doi:10.1016/j.jalgor.2004.12.003}}.

\bibitem{YangDA09swee}
Boting Yang, Danny Dyer, and Brian Alspach.
\newblock Sweeping graphs with large clique number.
\newblock {\em Discrete Mathematics}, 309(18):5770--5780, 2009.
\newblock \href {http://dx.doi.org/10.1016/j.disc.2008.05.033}
  {\path{doi:10.1016/j.disc.2008.05.033}}.

\end{thebibliography}
\end{document}